\documentclass[lettersize, journal]{IEEEtran}
\IEEEoverridecommandlockouts
\usepackage{amsmath,amssymb,amsfonts}
\usepackage{amsthm}
\usepackage{dsfont}
\usepackage{graphicx}
\usepackage{textcomp}
\usepackage{xcolor}
\usepackage{setspace,float,pstool,lettrine,newclude}
\usepackage{latexsym,algpseudocode,gensymb,balance,multirow,bm,wasysym}
\usepackage[colorlinks=true,bookmarks=false,citecolor=blue,urlcolor=blue]{hyperref}
\usepackage{cite}
\usepackage{makecell}
\ifCLASSOPTIONcompsoc
\usepackage[caption=false,font=normalsize,labelfont=sf,textfont=sf]{subfig}
\else
\usepackage[caption=false,font=footnotesize]{subfig}
\fi
\usepackage[linesnumbered,ruled,vlined]{algorithm2e}
\SetKwInput{KwInput}{Input}  
\SetKwInput{KwOutput}{Output}
\SetKw{KwBy}{by}
\makeatletter
\newcommand{\nosemic}{\renewcommand{\@endalgocfline}{\relax}}% Drop semi-colon ;
\newcommand{\dosemic}{\renewcommand{\@endalgocfline}{\algocf@endline}}% Reinstate semi-colon ;
\let\oldnl\nl% Store \nl in \oldnl
\newcommand{\nonl}{\renewcommand{\nl}{\let\nl\oldnl}}
\makeatother
\allowdisplaybreaks

\newtheorem{proposition}{Proposition}
\newtheorem{corollary}{Corollary}

\begin{document}

%\title{Query Scheduling for Goal-Oriented Effective Communication: Which Updates to Pull?}

\title{Pull-Based Query Scheduling for Goal-Oriented Semantic Communication}

\author{Pouya Agheli, 
\IEEEmembership{Member, IEEE}, Nikolaos Pappas, 
\IEEEmembership{Senior Member, IEEE}, and Marios Kountouris, \IEEEmembership{Fellow, IEEE}.
\thanks{P. Agheli and M. Kountouris are with the Communication Systems Dept., EURECOM, France, email: \texttt{agheli@eurecom.fr}. M. Kountouris is also with the Dept. of Computer Science and Artificial Intelligence, University of Granada, Spain, email: \texttt{mariosk@ugr.es}. N. Pappas is with the Dept. of Computer and Information Science, Linköping University, Sweden, email: \texttt{nikolaos.pappas@liu.se}. %The work of P. Agheli and M. Kountouris has received funding from the European Research Council (ERC) under the European Union’s Horizon 2020 research and innovation programme (Grant agreement No. 101003431). The work of N. Pappas has been supported by the Swedish Research Council (VR), ELLIIT, and the European Union (ETHER, 101096526).}
}}

\maketitle

\begin{abstract}
This paper addresses query scheduling for goal-oriented semantic communication in pull-based status update systems. We consider a system where multiple sensing agents (SAs) observe a source characterized by various attributes and provide updates to multiple actuation agents (AAs), which act upon the received information to fulfill their heterogeneous goals at the endpoint. A hub serves as an intermediary, querying the SAs for updates on observed attributes and maintaining a knowledge base, which is then broadcast to the AAs. The AAs leverage the knowledge to perform their actions effectively. To quantify the semantic value of updates, we introduce a \emph{grade of effectiveness} (GoE) metric. Furthermore, we integrate \emph{cumulative perspective theory} (CPT) into the long-term effectiveness analysis to account for risk awareness and loss aversion in the system. Leveraging this framework, we compute effect-aware scheduling policies aimed at maximizing the expected discounted sum of CPT-based total GoE provided by the transmitted updates while complying with a given query cost constraint. To achieve this, we propose a \emph{model-based} solution based on dynamic programming and \emph{model-free} solutions employing state-of-the-art deep reinforcement learning (DRL) algorithms. Our findings demonstrate that effect-aware scheduling significantly enhances the effectiveness of communicated updates compared to benchmark scheduling methods, particularly in settings with stringent cost constraints where optimal query scheduling is vital for system performance and overall effectiveness.
\end{abstract}


\section{Introduction}
Effectiveness remains a critical component in the rapidly evolving realm of communication and connectivity technologies, encompassing areas such as digital healthcare and digital twins. While efficiency often takes precedence, addressing effectiveness in processing, computing, and conveying information is equally important. Addressing this dimension enables us to tackle critical questions, such as: ``Which information needs to be communicated?'' and ``When is the right time to communicate it?''. 
In this context, communication effectiveness has been explored within the framework of the goal-oriented and semantic communication paradigms, where information is generated and transmitted \emph{only when} it has the potential to achieve the \emph{desired effect} or create the \emph{intended impact} at the endpoint to fulfill a specific goal \cite{kountouris2021semantics,popovski2020semantic,Strinati2024}. This concept significantly contributes to system scalability and efficient resource utilization by eliminating the need to acquire, process, and transmit information that turns out to be ineffective, irrelevant, or unnecessary.

The \emph{pull-based} communication model has significantly advanced the integration of effectiveness in status update systems. Unlike the traditional push-based model, where updates are sent at the source's discretion, the pull-based model allows the source to provide information in response to \emph{queries} from the endpoint, which controls the timing and nature of updates \cite{pullB, pullC, pullD, pullDD, pullE}. Combined with an \emph{effect-aware} querying policy, this approach can improve the effectiveness of communicated updates by an average of up to $149\%$ \cite{agheli2023effective}, leveraging timely knowledge of the source's evolution and channel conditions at the endpoint. As a result, the pull-based update model is considered a promising strategy for enhancing effectiveness in status update systems. However, scenarios involving multiple sources or a single source with multiple attributes remain underexplored from an effectiveness perspective. In such cases, the primary challenge lies in \emph{scheduling} updates to optimize query timing and achieve sustained high effectiveness over time. Addressing this challenge is essential for realizing the full potential of effectiveness in these complex systems.

When analyzing long-term effectiveness, the axioms of classical expected utility theory (EUT), where the weight of an occurrence is determined by its probability, could be challenged in the context of risk-aware semantic decision-making. Decisions involving risky prospects, perceptual context evaluation, and semantic utility valuation reveal several effects that the utility theory fails to capture. One prominent effect is \emph{certainty}, which suggests that merely probable outcomes are undervalued compared to certain ones. For example, in human-centric applications, there is often a preference to guarantee achieving a significant portion of a goal rather than risk achieving the entire goal with less certainty. Decision-making involving human agents and human-in-the-loop control of critical applications tends to be more conservative, prioritizing safety and performance guarantees over risk-taking, even when the latter has the potential to maximize the utility function. \emph{Cumulative prospect theory} (CPT) was developed as an alternative framework to EUT, addressing subjectivity and behavioral effects in decision-making \cite{kahneman1984choices, luce1991rank, tversky1992advances}. This theory offers a more refined and nuanced framework that captures how individuals perceive gains and losses relative to a reference point. It maps the outcomes to utility via a \emph{value function} and transforms the perceived probabilities of those outcomes into \emph{decision weights} using a \emph{weighting function}. These elements allow CPT to account for behavioral phenomena such as loss aversion and the overweighting of rare events. This makes CPT especially relevant in scenarios where infrequent but impactful updates matter, such as semantic communications, which prioritize information based on perceived significance rather than mere frequency. Consequently, CPT provides a more realistic basis for modeling decision-making under uncertainty.

In this paper, we investigate a pull-based end-to-end status update system in which multiple sensing agents monitor a source and provide updates to a group of actuation agents, each taking actions and working toward achieving \emph{heterogeneous} goals at the endpoint. The source comprises multiple attributes, and updates about these attributes are routed through a hub that dynamically determines which specific attribute to query at any given time. 
To assess the effectiveness of updates, we utilize the \emph{grade of effectiveness} concept introduced in \cite{agheli2024integrated}, incorporating two key semantic metrics: \emph{freshness} and \emph{usefulness}. Freshness quantifies how quickly an update becomes outdated after successfully reaching its endpoint. Usefulness, on the other hand, captures the update's importance at the endpoint, determined by its relevance to the specific goal. Consequently, the magnitude of an update's impact depends on the extent to which it satisfies the requirements of the goal. 

We formulate a scheduling problem to maximize the expected discounted sum of the total grade of effectiveness of communicated updates, subject to a query cost constraint. In this context, we depart from the classical EUT and leverage CPT to define long-term effectiveness. By solving this scheduling problem, we derive a class of effect-aware query scheduling policies for the hub. To this end, we propose two distinct approaches: \textit{(i)} A \emph{model-based} solution utilizing an iterative algorithm grounded in dynamic programming, and \textit{(ii)} \emph{model-free} solutions leveraging learning-based iterative algorithms, specifically using a range of well-established deep reinforcement learning (DRL) methods. We evaluate the performance of our proposed solutions against benchmark approaches through simulations, focusing on their effectiveness. The results show that effect-aware scheduling significantly enhances both the long-term effectiveness and the system's reliability, ensuring a minimum grade of effectiveness. In addition, the findings reveal that effect-aware scheduling delivers substantial performance improvements in scenarios with stringent cost constraints. Additionally, model-free approaches exhibit superior scalability compared to the model-based approach, making them more suitable for complex, real-world applications.

\subsection{Related Works}
Our work falls within the realm of query scheduling in pull-based communication systems. Although the pull-based model has been thoroughly explored, the scheduling problem aimed at enhancing the goal-oriented effectiveness of communicated updates under this model remains relatively unexplored. Existing studies on this problem, as highlighted in the literature, focus mainly on achieving accurate estimates of the source’s state over time \cite{chiariotti2022scheduling, holm2023goal, bui2023scheduling, cao2023goal, raghuwanshi2024goal}. These studies typically employ various forms of Kalman filters, relying on prior knowledge of the source’s dynamics at the estimator. The state of the source in these works evolves according to linear or nonlinear dynamics, with the estimation error typically quantified using mean square error (MSE). The MSE is widely used as a fundamental metric for evaluating communication effectiveness and serves as the basis for formulating the scheduling problem in these approaches.

In \cite{chiariotti2022scheduling} and \cite{holm2023goal}, the authors address the sensor scheduling problem with the aim of reducing estimation error in query responses using \emph{value of information} (VoI) metrics. These works were extended in \cite{bui2023scheduling, cao2023goal, raghuwanshi2024goal}, where energy efficiency was incorporated as an additional maximization objective alongside minimizing state estimation error. Specifically, \cite{bui2023scheduling} sought to capture VoI by accounting for the sensors' observation noise and channel conditions, while \cite{cao2023goal} assessed the importance of updates based on their binary information value within a two-state environment. However, none of these works has yet considered the freshness of updates, which contrasts with solving scheduling or queuing problems aimed at minimizing the \emph{age of information} (AoI), which quantifies freshness, along with its variants \cite{AoIA, AoIB, AoIC, AoID, AoIE, stamatakis2023optimizing, delfani2024semantics, delfani2024optimizing}. Some studies, such as in \cite{ayan2019age}, have sought to balance VoI and AoI by developing joint schedulers to optimize this trade-off. Despite these advances, significant gaps remain in exploring how query scheduling can enhance effectiveness when managing systems with multiple attributes. 

The query scheduling problem we address is built on the grade of effectiveness metric, incorporating cumulative prospect theory for long-term effectiveness analysis. This framework provides a more comprehensive model that accounts for multiple attributes and integrates risk awareness. By enabling the evaluation of uncertain outcomes with subjective decision-making aspects, it is particularly suited for scenarios where attributes such as information freshness and usefulness must be considered simultaneously under risk conditions rather than optimizing a single attribute in isolation. In \cite{agheli2024access}, we proposed a self-decision multiple access scheme to address a \emph{risk-agnostic} variant of the effectiveness enhancement problem. In that framework, sensing agents autonomously decide whether and when to provide updates in response to incoming queries, guided by a predefined objective. While the overarching goal aligns with the one pursued here, the approach in \cite{agheli2024access} is decentralized and agent-centric, in contrast to the centralized strategy proposed in this work. As such, the two methodologies address complementary layers of the decision-making process, offering a more holistic perspective when considered together. Furthermore, the form of the grade of effectiveness metric in this context differs, leading to a distinct objective function and a unique approach to address the problem.

\subsection{Contributions}
The main contributions of this work can be summarized as follows.
\begin{itemize}
    \item We develop effect-aware query scheduling policies that account for the impact of updates and aim to maximize the expected discounted sum of the total effectiveness grade of communicated updates. This metric is grounded in cumulative prospect theory, incorporating risk awareness into the update scheduling process. The proposed policies are designed to comply with query cost constraints over time. To achieve this, we propose a model-based solution that leverages an iterative algorithm rooted in the dynamic programming approach. 
    
    \item To overcome the complexity and scalability limitations of the model-based approach, we extend our work to include model-free solutions for achieving effect-aware query scheduling policies. This involves leveraging learning-based iterative algorithms and implementing three prominent deep reinforcement learning methods. 

    \item We evaluate the performance of the proposed model-based and model-free solutions through simulations and compare them against benchmark scheduling methods. Our results show that effect-aware scheduling either improves effectiveness compared to the benchmarks with the same number of queries sent or achieves comparable effectiveness while significantly reducing the number of communicated updates, depending on the specific effect-aware approach employed. Additionally, effect-aware scheduling further boosts effectiveness in scenarios with stringent cost constraints. In terms of scalability, the model-free solutions outperform the model-based approach, delivering superior performance in larger and more complex scenarios.
\end{itemize}

\textit{Notations:} $\mathbb{R}$, $\mathbb{R}_0^+$, and $\mathbb{N}$ represent the sets of real, non-negative real, and natural numbers, respectively. $\mathbf{1}_x$ denotes an all-ones column vector of size $x$. $\mathbb{E}[\cdot]$ is the expectation operator, $\lvert\cdot \rvert$ indicates the absolute value operator, $\lceil\cdot\rceil$ is the ceiling operator, and $\mathcal{O}(\cdot)$ describes a function's growth rate. 

\section{System Model}\label{sec2}
We consider a time-slotted end-to-end status update system in which $N$ \emph{sensing agents} (SAs) observe a time-variant source and transmit noisy observations in the form of update chunks to a \emph{hub} over a shared medium (see Fig.~\ref{Fig:sys_mod}). The source is characterized by $M$ attributes, where the $m$-th, for $m=1,\dots, M$, attribute is denoted by $x_m(t)$ at time slot $t\in \mathbb{N}$. Each attribute is modeled as an independent and identically distributed (i.i.d.) random variable, where the $m$-th attribute takes values from a finite set $\mathcal{X}_m=\{i~\lvert~i=1,\dots,\lvert\mathcal{X}_m\rvert\}$ with the $i$-th element occurring with probability $P_m(x_i)$ and $P_m(\cdot)$ representing the probability mass function (pmf) of that attribute.
% belongs to a finite set $\mathcal{X}_m=\{i~\lvert~i=1,\dots,\lvert\mathcal{X}_m\rvert\}$, where each element occurs with a specific probability.
Upon receiving a query regarding the $m$-th attribute from the hub,
the $n$-th SA, where $n=1, \dots, N$, observes that attribute, generates and sends an observation denoted as $\hat{x}_{nm}(t)\in\mathcal{X}_m$. The observation is correct, i.e., $\hat{x}_{nm}(t)=x_m(t)$, with a probability of $p_{{\rm{o}},nm}\in [0,1)$.
\begin{figure}[t!]
    \centering
    \includegraphics[width=0.35\textwidth]{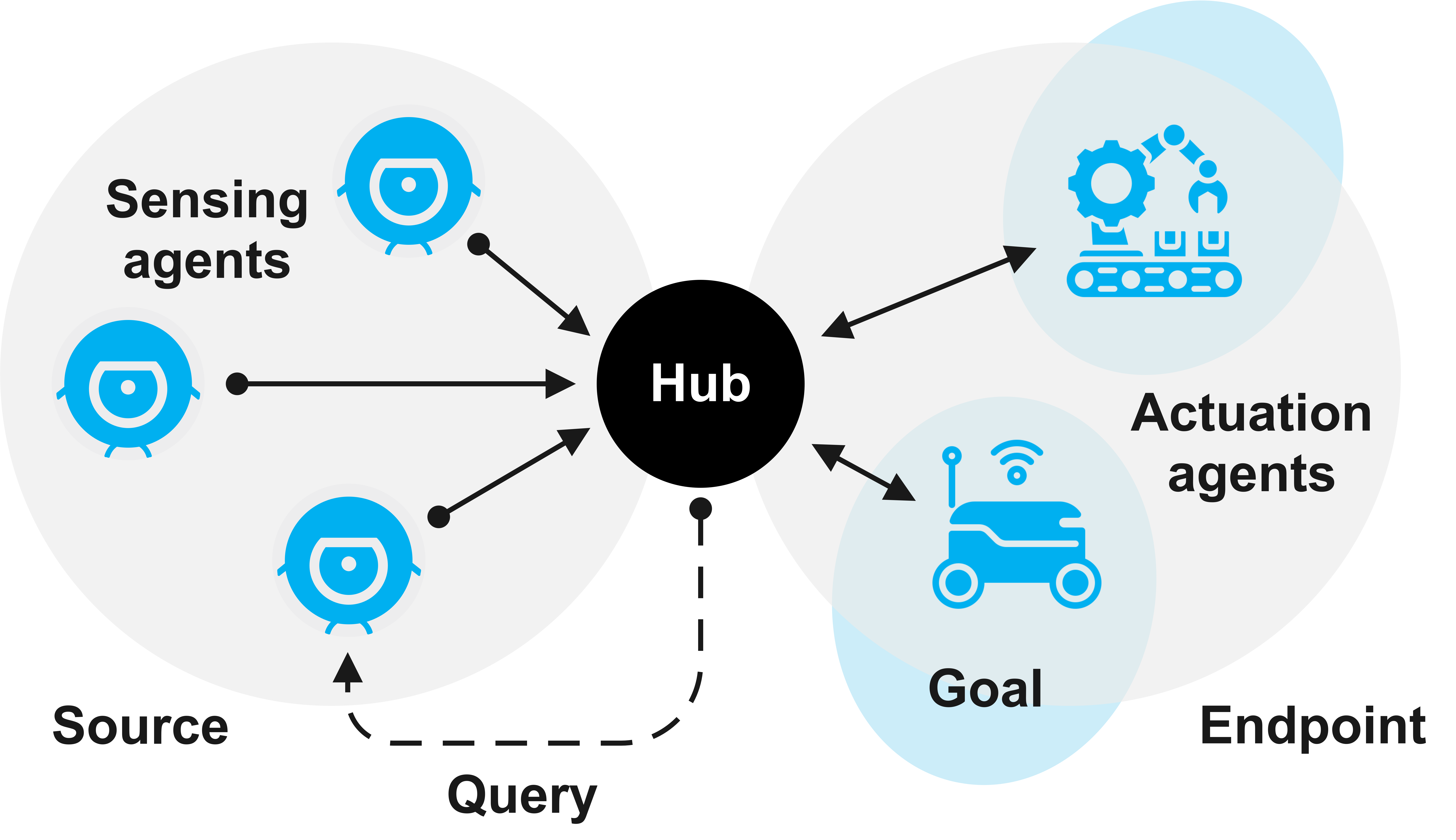}
    \caption{End-to-end update communication to satisfy heterogeneous goals.}
    \label{Fig:sys_mod}
\end{figure}

At the endpoint, $K$ \emph{actuation agents} (AAs) operate to accomplish their heterogeneous goals. To do so, the AAs require access to a \emph{knowledge base}, represented as $\mathbf{y}(t) = [y_1(t), \dots, y_M(t)]^T$, which is provided and broadcast by the hub during the $t$-th time slot. The $m$-th element of the vector, $y_m(t)\in \mathcal{X}_m$, contains the latest update on the $m$-th attribute. However, the $k$-th AA may require only a subset of attributes, denoted by $\mathcal{M}_k$, where $\lvert\mathcal{M}_k\rvert\leq M$ for $k=1,\dots, K$, to perform its actions. In this framework, the update (communication) channels between the hub and the AAs, as well as the query channels, are assumed to be error-free and zero-latency.\footnote{These assumptions enable us to focus on query scheduling without being overshadowed by secondary effects. They are justified in scenarios with high-speed, reliable downlink connections, where errors and delays are negligible.} Nevertheless, update packets (or packet chunks) sent from the $n$-th SA to the hub are subject to erasure with a probability $p_{{\rm{e}}, n}\in (0,1]$. 
% hence merely looking up the corresponding elements of $\mathbf{y}(t)$.

\subsection{Scheduling Queries}\label{sec2:partA}
The hub schedules queries about attributes to update $\mathbf{y}(t)$ in response to the needs of the AAs. In each time slot, the hub selects at most \emph{one}\footnote{This condition simplifies the presentation of the core idea, ensuring clarity and focus while minimizing unnecessary complexity.} attribute and queries an SA about that specific attribute. However, the hub operates without full knowledge of the evolution of the source. The query indicator for the $m$-th attribute from the $n$-th SA at the $t$-th slot is denoted by $\alpha_{nm}(t) = \{0, 1\}$, where $\alpha_{nm}(t) = 1$ indicates the query arrival, and $\alpha_{nm}(t) = 0$ otherwise. Based on this, we can define $\mathbf{A}(t)=[\mathbf{a}_{1}(t),\dots,\mathbf{a}_{M}(t)]$ as the $N\times M$ query matrix, with each column $\mathbf{a}_m(t)=[\alpha_{1m}(t),\dots,\alpha_{Nm}(t)]^T$ representing an $N\times1$ query vector corresponding on the $m$-th attribute. Therefore, the $m$-th element of $\mathbf{y}(t)$ is obtained as follows:
\begin{align}
    y_m(t) = 
    \begin{cases}
         y_m(t-1), ~\mathbf{a}^T_m(t) (\mathbf{1}_N - \mathbf{e}(t)) = 0;\\
         \mathbf{a}^T_m(t)\mathbf{\hat{x}}_m(t), ~\text{otherwise},
    \end{cases}
\end{align}
where $y_m(0)=0, \forall m$, $\mathbf{\hat{x}}_m(t)=[\hat{x}_{1m}(t),\dots, \hat{x}_{Nm}(t)]^T$, and $\mathbf{e}(t)=[e_{1}(t),\dots, e_{N}(t)]^T$ being an $N\times 1$ error matrix in which $e_{n}(t)\in\{0, 1\}$ is the error indicator for the update sent from the $n$-th SA. Here, $e_{n}(t)=1$ indicates that an error occurs due to channel erasure; otherwise, $e_{n}(t)=0$.

For scheduling, every SA, whether existing or newly joined in the network, say the $n$-th one, shares its observation vector $\mathbf{p}_{{\rm{o}}, n}=[p_{{\rm{o}},n1},\dots,p_{{\rm{o}}, nM}]^T \in [0,1]^{M\times 1}$ with the hub during designated slots, while each element is subject to erasure with a probability $p_{{\rm{e}}, n}$. Additionally, we consider that scheduling, query arrival, generation of updates, and subsequent transmission occur within a \emph{single} slot.

\subsection{Grade of Effectiveness Definition}
The \emph{grade of effectiveness} (GoE) of an update on the $m$-th attribute is denoted by $\operatorname{GoE}_m(t)\in\mathbb{R}_0^+$ and is modeled as a composite function $f:\mathbb{R}_0^{+} \times\mathbb{R}_0^{+}\rightarrow \mathbb{R}_0^+$ of that update's freshness and usefulness at the $t$-th slot. Inspired from \cite[Section~III]{agheli2023semantic} and \cite[Section~III-A]{agheli2024integrated}, we can write
\begin{equation}\label{eq:GoE}
    \operatorname{GoE}_m(t) = f\big(f_\Delta(\Delta_m(t)), f_u(u_m(t))\big)
\end{equation}
where $f_\Delta : \mathbb{N} \rightarrow \mathbb{R}_0^+$ and $f_u : \mathbb{R}_0^+ \rightarrow \mathbb{R}_0^+$ denote non-increasing penalty and non-decreasing utility functions, respectively. Also, $\Delta_m(t)\in\mathbb{N}$ indicates the AoI for the $m$-th attribute at time slot $t$ that quantifies the freshness of information and is defined as follows: 
\begin{equation}
    \Delta_{m}(t) = t - \operatorname{max}\left\{t^\prime~|~t^\prime \leq t, \mathbf{a}^T_m(t^\prime)\mathbf{\hat{x}}_m(t^\prime) = x_m(t^\prime)\right\}
\end{equation}
where $\Delta_m(0)=1, \forall m$.

Moreover, $u_m(t)$ represents the usefulness of the $m$-th attribute, which takes a value from a finite set $\mathcal{U}=\{\nu_j~\lvert~j=1,\dots,\lvert\mathcal{U}\rvert\}$. Without loss of generality, we assume that the hub utilizes a surjective mapping function $g_m:\mathbb{N}\times \mathcal{X}_m \rightarrow \mathcal{U}$ for the $m$-th attribute, which maps the content of an update for the attribute to its overall usefulness. The mapping accounts for the time-variant goal requirements and the (potentially non-uniform) importance weights assigned by the AAs. Thus, we can define $u_m(t) = g_m(t; y_m(t))$ subject to $y_m(t)=x_m(t)$. If $y_m(t)\neq x_m(t)$, the update is not useful, i.e., $u_m(t)=0$. We assume $u_m(0) =  g_m(0; 0) = \nu_1, \forall m$. Given $P_{\nu, m}(\cdot)$ as the pmf over $\mathcal{U}$ for the $m$-th attribute, $u_m(t)$ could be modeled as an i.i.d. random variable where $u_m(t)=\nu_j, \forall \nu_j\in \mathcal{U}$, with the probability of $P_{\nu, m}(\nu_j)$ given by
\begin{equation}
    P_{\nu, m}(\nu_j)=\sum_{i:\,x_i=g_m^{-1}(t;\nu_j)} P_m(x_i)
\end{equation}
at the $t$-th time slot, where $P_m(x_i)$ is the occurrence probability of the $i$-th element of $\mathcal{X}_m$. 

\section{Query Scheduling Problem}
Through this section, we begin by formulating a scheduling problem for querying attribute updates. Then, we recast this problem as a \emph{constrained Markov decision problem} (CMDP) to be solved. 

Before formulating the problem, we introduce a total GoE for the combined set of all required attributes as follows:
\begin{align}\label{eq:AvgGoE}
    \operatorname{GoE}(t) =
    % \dfrac{1}{\left\lvert \bigcup_{k=1}^{K} \mathcal{M}_k \right\lvert} 
    \sum_{m \in \bigcup_{k=1}^{K} \mathcal{M}_k} \operatorname{GoE}_m(t)
\end{align}
at the $t$-th slot, given $\mathcal{M}_k$ being the subset of the attributes required by the $k$-th AA, where $k=1,\dots,K$. 

\subsection{Problem Formulation}\label{sec3:partA}
The objective is to maximize the expected discounted sum of the CPT-based total GoE by optimally scheduling queries while ensuring that the expected discounted cumulative query cost does not exceed a cost constraint $C_{\rm max}$. A class of \emph{effect-aware} query scheduling policies, notated as $\pi^*$, is derived by solving the following problem:
\begin{align}\label{eq:opt-I}
    \mathcal{P} :~ &\underset{\pi}{\operatorname{max}} ~~ \mathbb{E}_{\pi}\!\left[\,\sum_{t = 0}^{\infty} \gamma^t v_{\rm cpt} \big(\!\operatorname{GoE}(t)\big) \,\Big|\operatorname{GoE}(0), \pi \right] \nonumber \\
    & {\rm s.t.} ~~ \mathbb{E}_{\pi}\!\left[\,\sum_{t = 0}^{\infty} \gamma^t v^+_{\rm cpt}\big(f_c(\mathbf{1}_N^T \mathbf{A}(t) \mathbf{1}_M)\big) \,\Big|\,\pi \right] \leq C_{\rm max}, \nonumber \\
    &~~~~~\,  \mathbf{1}_N^T \mathbf{A}(t) \mathbf{1}_M \leq 1, \forall t,
\end{align}
where the expectations are taken over CPT-weighted probabilities, and $\gamma \in [0,1)$ is a discount factor. Also, $f_c:\{0,1\} \rightarrow \mathbb{R}_0^+$ is a non-increasing function to quantify the induced query cost.

The function $v_{\rm cpt}: \mathbb{R}_0^+\rightarrow \mathbb{R}$ represents the standard two-part CPT-based value function \cite{tversky1992advances}, which is concave for gains, convex for losses, and exhibits greater sensitivity to losses than to gains, reflecting loss aversion. The value function for gains (losses) is denoted by $v^+_{\rm cpt}(\cdot)$ ($v^-_{\rm cpt}(\cdot)$) and is monotonically non-decreasing (non-increasing) under positive (negative) prospects only. Otherwise, it evaluates to zero. Let $\operatorname{GoE}_{\rm ref} \in \mathbb{R}_0^+$ denote the \emph{reference point} for the total GoE. In this context, any instance where $\operatorname{GoE}(t) \geq \operatorname{GoE}_{\rm ref}, \forall t$, is considered a gain, while $\operatorname{GoE}(t) < \operatorname{GoE}_{\rm ref}$ refers to losses. For the induced query cost, we set the reference point to zero, making the prospect associated with the cost strictly positive.

\subsection{CMDP Modeling}\label{sec3:partB}
We can transform $\mathcal{P}$ into an \emph{infinite-horizon} CMDP problem by defining its components in the following.

\textbf{States.} The state space of the process is denoted by $\mathcal{S}$. The state at the $t$-th time slot, identified as $s(t)\in\mathcal{S}$, is represented as a tuple containing the AoI and the usefulness of updates for all attributes required by the AAs, as below:
\begin{equation*}
    s(t) = \big(\Delta_1(t), \dots, \Delta_{\bigcup_{k=1}^{K} \mathcal{M}_k}(t), u_1(t), \dots, u_{\bigcup_{k=1}^{K} \mathcal{M}_k}(t)\big).
\end{equation*}
Without loss of generality, we assume that the maximum acceptable AoI is $\Delta_{\rm max}$, i.e., $\Delta_m(t)\leq\Delta_{\rm max}, \forall m,t$. Given this, the size of the state space, denoted as $\lvert \mathcal{S}\rvert$, is determined as follows:
\begin{equation*}
    \lvert \mathcal{S}\rvert = \big\lvert\Delta_{\rm max} \lvert\mathcal{U}\rvert \big\rvert^ {\left\lvert\bigcup_{k=1}^{K} \mathcal{M}_k\right\rvert}.
\end{equation*}

\textbf{Actions.}
At the $t$-th slot, the action is denoted by $a(t)$, which is taken from the action space $\mathcal{A}=\{0\}\cup\bigcup_{k=1}^{K} \mathcal{M}_k$ having the size of $|\mathcal{A}| = \lvert \bigcup_{k=1}^{K} \mathcal{M}_k\rvert + 1$. 
In this definition, $a(t)=m, \forall m\neq0$, indicates querying the $m$-th attribute, i.e., $\mathbf{1}_N^T \mathbf{a}_m(t)=1$ (see Section~\ref{sec2:partA}), and $a(t)=0$ indicates that \emph{no} query is made during the $t$-th slot, i.e., $\mathbf{1}_N^T \mathbf{A}(t) \mathbf{1}_M=0$. Once the $m$-th attribute is selected for querying, the specific $n_m$-th SA to be queried about that attribute is determined as
\begin{equation}\label{eq:SelSA}
    n_m = \underset{n=1,\dots,N}{\operatorname{arg\,max}}~  (1 - p_{{\rm{e}}, n}) p_{{\rm{o}},nm}
\end{equation}
where $p_{{\rm{e}}, n}$ is the channel erasure probability between the $n$-th SA and the hub, and $p_{{\rm{o}},nm}$ is the likelihood of the $n$-th SA observing the $m$-th attribute correctly (see Section~\ref{sec2}). Applying \eqref{eq:SelSA}, we reach $\alpha_{n_mm}(t)=1, \forall m\in\mathcal{A}\setminus\{0\}$, for the $m$-th attribute that is selected during the $t$-th slot. Otherwise, we have $\alpha_{nm}(t)=0, \forall n,m$, satisfying $\mathbf{1}_N^T \mathbf{A}(t) \mathbf{1}_M \leq 1$.

% The action space is defined as $\mathcal{A}=\{(m, n)~|~m= 1,\dots,M, n=1,\dots,N\}\cup\{0\}$. The tuple $(m, n), \forall n,m\neq0$, implies querying the $m$-th attribute from the $n$-th SA, i.e., $\alpha_{nm}(t)=1$ at the $t$-th slot (see Section~\ref{sec2:partA}). However, the outcome $0$ denotes not querying any attribute, i.e., $\alpha_{nm}(t)=0, \forall n, m$. The size of the action space is $|\mathcal{A}| = N\times M + 1$.

\textbf{Transition probabilities.} 
The transition probability of reaching state $s(t+1)$ at the $t+1$-th slot from state $s(t)$ by taking action $a(t)$ is denoted as $p(s(t+1)|s(t), a(t))$. Hence, these transition probabilities are derived as follows:
\begin{itemize}
\renewcommand{\labelitemi}{\scriptsize\(\blacksquare\)}
    \item If $s(t+1)$ shows a successful update of the $m$-th attribute, i.e., $\Delta_m(t+1)=1$ and $u_m(t+1)=\nu_j, \forall \nu_j\in\mathcal{U}, m\in\bigcup_{k=1}^{K} \mathcal{M}_k$, whereas $\Delta_{m^\prime}(t+1)=\Delta_{m^\prime}(t)+1$ and $u_{m^\prime}(t+1)=u_{m^\prime}(t), \forall {m^\prime}\neq m$, we have
    \begin{align*}
        &p(s(t+1)|s(t), a(t)=m) = \\ 
        &~~~~~~~~~~P_{\nu, m}(\nu_j)(1 - p_{{\rm{e}}, n_m}) p_{{\rm{o}},n_mm}.
    \end{align*}
    \item Else if $\Delta_m(t+1)=\Delta_m(t)+1$ and $u_m(t+1)=u_m(t), \forall m\in\bigcup_{k=1}^{K} \mathcal{M}_k$, we reach 
    \begin{align*}
        &p(s(t+1)|s(t), a(t)=m) = \\
        &~~~~~~~~~~(1 - p_{{\rm{o}},n_mm}) +  p_{{\rm{e}}, n_m} p_{{\rm{o}},n_mm},
    \end{align*}
and $p(s(t+1)|s(t), a(t)=0) = 1$. 
\end{itemize}
All other probabilities are equal to zero.

% Hence, we obtain the transition probabilities as below:
% \begin{itemize}
%     \renewcommand{\labelitemi}{\scriptsize\(\blacksquare\)}
%     \item If $s(t+1)=(\Delta_m(t+1)=1, u_m(t+1)=\nu_j; \Delta_{m^\prime}(t+1)=\Delta_{m^\prime}(t)+1, u_{m^\prime}(t+1)=u_{m^\prime}(t), \forall {m^\prime}\neq m), \forall \nu_j\in\mathcal{U}$, we have
%     \begin{align*}
%         &p(s(t+1)|s(t), a(t)=m) =  \\
%         &~~~~~~~~~P_{\nu, m}(\nu_j)(1 - p_{{\rm{e}}, n_m}) p_{{\rm{o}},n_mm}, \forall m\in\bigcup_{k=1}^{K} \mathcal{M}_k.
%     \end{align*}
%     \item Else if $s(t+1)=(\Delta_{m^\prime}(t+1)=\Delta_{m^\prime}(t)+1, u_{m^\prime}(t+1)=u_{m^\prime}(t), \forall m^\prime)$, we can write
%     \begin{align*}
%         &p(s(t+1)|s(t), a(t)=m) =\\
%         &~~~~~~~~~(1 - p_{{\rm{o}},n_mm}) +  p_{{\rm{e}}, n_m} p_{{\rm{o}},n_mm}, \forall m\in\bigcup_{k=1}^{K} \mathcal{M}_k, \\
%         &p(s(t+1)|s(t), a(t)=1) = 1.
%     \end{align*}
% \end{itemize}

Letting $w_{\rm cpt}:\{0, 1\}\rightarrow\{0, 1\}$ represent the CPT-based weighting function \cite{tversky1992advances}, the CPT-weighted transition probability is given by $w_{\rm cpt}(p(s(t+1)|s(t), a(t)))$. The weighting function is continuous, non-decreasing, and typically exhibits an inverse S-shaped curve. This means that it is concave for low probabilities and convex for high probabilities, with boundary conditions $w_{\rm cpt}(0)=0$ and $w_{\rm cpt}(1)=1$.

\textbf{Rewards.}
The immediate reward for transitioning from state $s(t)$ to state $s(t+1)$ upon taking action $a(t)$ is defined as $r(s(t), a(t), s(t+1)) = v_{\rm cpt}(\operatorname{GoE}(t+1))$.

\begin{proposition}\label{prop:prop1}
The modeled CMDP satisfies the weak accessibility condition.
\end{proposition}
\begin{proof}
We can decompose the state space into $M+1$ disjoint subsets such that $\mathcal{S}=\bigcup_{m=0}^{M}\mathcal{S}_m$. For each $m=1, \dots, M$, the subset $\mathcal{S}_m$ corresponds to the set of states where $\Delta_m(t) < \Delta_{\rm max}$, while $\Delta_{m^\prime}(t)=\Delta_{\rm max}, \forall {m^\prime}\neq m$, and $u_{m^\prime}(t)$ remains constant within this subset. Under the policy $\mathbf{1}_N^T \mathbf{a}_m(t)=1$, the subset $\mathcal{S}_m$ indicates the set of \emph{recurrent} states. The subset $\mathcal{S}_0$ consists of the remaining states, all of which are \emph{transient} under some policy. Therefore, the weak accessibility condition holds \cite[Definition~4.2.2]{bertsekas2007volume}, as the process can reach any recurrent state from a transient one with a proper policy.
\end{proof}
\begin{corollary}\label{cor:cor_1}
With the weak accessibility condition satisfied for the defined CMDP model, the expected discounted sum of the CPT-based total GoE in $\mathcal{P}$ remains invariant across all initial states \cite[Proposition~4.2.3]{bertsekas2007volume}. Consequently, $v_{\rm cpt} \big(\!\operatorname{GoE}(t)\big)$, is independent of $\operatorname{GoE}(0)$, for all $t\geq 1$.
\end{corollary}
\begin{corollary}\label{cor:cor_2}
Satisfying the weak accessibility condition guarantees the existence of a class of stationary optimal policies, denoted as $\pi^*$, for $\mathcal{P}$ in \eqref{eq:opt-I}, which are unichain \cite[Proposition~4.2.6]{bertsekas2007volume}. 
\end{corollary}

\section{Model-Based Solution for Effect-Aware Scheduling}\label{sec4}
In this section, we explore the duality of the query scheduling problem formulated in Section~\ref{sec3:partA}. Subsequently, leveraging the CMDP model defined in Section~\ref{sec3:partB}, we propose a \emph{model-based} solution to derive effect-aware scheduling policies.

\subsection{Dual Problem}\label{sec4:partA}
By applying Corollary~\ref{cor:cor_1}, we incorporate the cost constraint of the scheduling problem $\mathcal{P}$ in \eqref{eq:opt-I} into its objective function and derive the Lagrangian function associated with the scheduling problem, as follows:
\begin{align}\label{eq:lag}
    &\mathcal{L}(\pi, \mu) =\nonumber \\
    &\mathbb{E}_{\pi}\!\left[\,\sum_{t = 0}^{\infty} \gamma^t \Big(v_{\rm cpt}\big(\!\operatorname{GoE}(t)\big) - \mu v^+_{\rm cpt}\big(f_c(\mathbf{1}_N^T \mathbf{A}(t) \mathbf{1}_M)\big)\Big) \,\Big|\,\pi \right] \nonumber \\
    &~~~ + \mu C_{\rm max}
\end{align}
where $\mu\geq0$ is the Lagrange multiplier. Based on the CMDP model, the second constraint of $\mathcal{P}$ is relaxed when formulating the dual problem.
Consequently, the dual problem is expressed as below:
\begin{align}\label{eq:opt-II}
    \widehat{\mathcal{P}} :~ &\underset{\mu}{\operatorname{inf}} ~\underset{\pi}{\operatorname{sup}} ~~ \mathcal{L}(\pi, \mu) \nonumber \\
    & {\rm s.t.} ~~ \mu \geq 0.
\end{align}

There exists a saddle point that ensures the convergence of the original scheduling problem, i.e., $\mathcal{P}$, and the dual one, i.e., $\widehat{\mathcal{P}}$, to the same values. This is due to \textit{(i)} the finite state space of the CMDP model, satisfying the growth condition \cite{altman1999constrained}, and \textit{(ii)} the boundness of the immediate rewards in that model. In this sense, the following equality holds \cite[Corollary~12.2]{altman1999constrained}:
\begin{align}
    \underset{\mu \geq 0}{\operatorname{inf}} ~ \underset{\pi}{\operatorname{sup}} ~ \mathcal{L}(\pi, \mu) = \underset{\mu \geq 0}{\operatorname{inf}} ~ \mathcal{L}(\pi^*, \mu) = \underset{\pi}{\operatorname{sup}} ~ \mathcal{L}(\pi, \mu^*) 
\end{align}
where $\mu^*$ denotes a non-negative optimal Lagrange multiplier for some $\pi$, which is feasible owing to \textit{(i)} and \textit{(ii)} under the Slater's condition \cite[Theorem~12.8]{altman1999constrained}. 

We can now develop an \emph{iterative algorithm} to solve the dual scheduling problem and derive the class of optimal policies.

\subsection{Iterative Algorithm}
We propose Algorithm~\ref{alg_1}, which utilizes \emph{inner} and \emph{outer} loops to iteratively compute the optimal policy $\pi^*$ and the optimal Lagrange multiplier $\mu^*$, leveraging the \emph{value iteration} method for policy optimization in the inner loop and the \emph{bisection search} method for refining $\mu^*$ in the outer loop. Given the interdependence between the policy and the Lagrange multiplier, the algorithm alternates between these loops and is executed iteratively until convergence to achieve both optimal values.
\SetKwFunction{FMain}{$\text{derive}\_\text{policy}$} 
% \SetKwFunction{FBrent}{$\text{derive}\_\text{multiplier}$} 
\SetKwProg{Fn}{Function}{:}{}
\begin{algorithm}[t!]
{
% \fontsize{8.8}{8.8}
\DontPrintSemicolon
    \caption{Computing $\pi^*$ and $\mu^*$} \label{alg_1}
    \KwInput{Given parameters $M$, $N$, $T\gg 0$, and $\gamma$, $C_{\rm max}$. CMDP's state space, i.e., $\mathcal{S}$, and action space, i.e., $\mathcal{A}$. 
    Shapes of $v_{\rm cpt}(\cdot)$, $w_{\rm cpt}(\cdot)$, and $f_c(\cdot)$. 
    Tolerance $\epsilon_\mu$. Mixing factor $\eta$.
    Initial values $l \leftarrow 0$, $\mu^{(0)} \leftarrow 0$, $\mu_{-}^{(0)} \leftarrow 0$, $\mu_{+}^{(0)} > 0$, $\pi^- \leftarrow 0$, and $\pi^+ \leftarrow 0$.
    }
    % \KwOutput{Optimal $\pi^*$ and $\mu^*$.}
    Initialize $\pi^*(s), \forall s \in \mathcal{S}$, via \FMain{$\mu^{(0)}$}.\\
    \lIf{$\mathbb{E}_{\pi}\!\left[\,\sum_{t = 0}^{T} \gamma^t v^+_{\rm cpt}(f_c(\mathbf{1}_N^T \mathbf{A}(t) \mathbf{1}_M))\right] \!\leq\! C_{\rm max}$}{\textbf{goto} {\scriptsize{\textbf{\ref{line:return1}}}}.}
        
    \nonl\Comment{\footnotesize{\textit{Outer loop (Bisection search)}}}
    
    \While{$\big\lvert\mu_{+}^{(l)}-\mu_{-}^{(l)}\big\rvert \geq \epsilon_\mu$}
    {
    \nonl\textit{\textbf{Step}} $l$: 
   
    set $l \leftarrow l+1$, $\mu_{-}^{(l)} \leftarrow \mu_{-}^{(l-1)}$, and $\mu_{+}^{(l)} \leftarrow \mu_{+}^{(l-1)}$.\\
    Update $\mu^{(l)} \leftarrow \frac{\mu_{-}^{(l)} + \mu_{+}^{(l)}}{2}$. \\
    Improve $\pi^* \leftarrow$ \FMain{$\mu^{(l)}$}.\\
    \If{$\mathbb{E}_{\pi}\!\left[\,\sum_{t = 0}^{T} \gamma^t v^+_{\rm cpt}(f_c(\mathbf{1}_N^T \mathbf{A}(t) \mathbf{1}_M))\right] \!\geq\! C_{\rm max}$}{$\mu_{-}^{(l)} \leftarrow \mu^{(l)}$, and $\pi^- \leftarrow$ \FMain{$\mu_{-}^{(l)}$}.}
    \lElse{$\mu_{+}^{(l)} \leftarrow \mu^{(l)}$, and $\pi^+ \leftarrow$ \FMain{$\mu_{+}^{(l)}$}.}
    
    } 
    \If{$\mathbb{E}_{\pi}\!\left[\,\sum_{t = 0}^{T} \gamma^t v^+_{\rm cpt}(f_c(\mathbf{1}_N^T \mathbf{A}(t) \mathbf{1}_M))\right] \!<\! C_{\rm max}$}{$\pi^*(s) \leftarrow \eta \pi^-(s) + (1-\eta)\pi^+(s), \forall s \in \mathcal{S}$.}
    \KwRet $\mu^* = \mu^{(l)}$ and $\pi^*(s), \forall s\in\mathcal{S}$.\label{line:return1}
    
  %   \nonl
  %   \nonl\hrulefill\\
  %   \nonl
  %   \Fn{\FBrent{$\pi^*$, $l$, $\mu_{-}^{(l)}$, $\mu_{+}^{(l)}$}}{
  %   \KwInput{Fetched history $\mu^{(l-3)}$, $\mu^{(l-2)}$, and $\mu^{(l-1)}$. Golden ratio $\phi$.}

  %   Initialize $\mu \leftarrow \frac{\mu_{-}^{(l)}+\mu_{+}^{(l)}}{2}$.
         
  %   \nonl\Comment{\footnotesize{\textit{Parabolic interpolation}}}
    
  %   \If{$\mu^{(l-3)} \neq \mu^{(l-2)} \neq \mu^{(l-1)}$}{update $\mu$ using \eqref{}.
    
  %   \lIf{$\mu \in [\mu_{-}^{(l)}, \mu_{+}^{(l)}]$ and $\big|\mu - \mu^{(l-1)}\big| < \frac{\mu_{+}^{(l)} - \mu_{-}^{(l)}}{2}$}{\textbf{goto}~{\scriptsize{\textbf{\ref{line:brent-out}}}}.} 
  %   }
  
  %   \nonl\Comment{\footnotesize{\textit{Golden section search}}}
    
  %   \If{$\mu \leq \mu^{(l-1)}$}{$\mu \leftarrow \mu^{(l-1)} - \frac{\phi-1}{\phi}(\mu^{(l-1)} - \mu_{-}^{(l)})$}
  %   \lElse{$\mu \leftarrow \mu^{(l-1)} + \frac{\phi-1}{\phi}(\mu_{+}^{(l)} - \mu^{(l-1)})$}
  %       \KwRet $\mu$\label{line:brent-out}.
  % }
  
    \nonl
    \nonl\hrulefill\\
    \nonl
    \Fn{\FMain{$\mu$}}{
    \vspace{0.07cm}
    \KwInput{Global parameters from the outer loop. Sensitivity criterion $\epsilon_\pi$. Initial values $i \leftarrow 1$, $\pi \leftarrow 0$, and $V^{(0)}_\pi \leftarrow 0$.}

    \nonl\Comment{\footnotesize{\textit{Inner loop (Value iteration)}}}
    
    \nonl\textit{\textbf{Iteration}} $i$: \\
    \For{state $s \in \mathcal{S}$\label{line:iter_t}}{compute $V^{(i)}_\pi(s)$ from \eqref{eq:V-func}--\eqref{eq:Q-func}\label{line:compute_v}.\\
    Improve $\pi(s)$ using \eqref{eq:Q-func}--\eqref{eq:opt-policy}.
    }
    \If{$\operatorname{sp}\!\big(V^{(i)}_\pi - V^{(i-1)}_\pi\big) \geq \epsilon_\pi$ as in \eqref{eq:sp-func}}{increment $i \leftarrow i+1$, and \textbf{goto} {\scriptsize{\textbf{\ref{line:iter_t}}}}.}
    \KwRet $\pi(s), \forall s\in\mathcal{S}$.
  }
}
\end{algorithm}

\subsubsection{Computing $\pi^*$}
Using the value iteration approach based on dynamic programming and bootstrapping future returns, the value function $V^{(i)}_\pi(s), \forall s\in\mathcal{S}$, within the $i$-th iteration, following the policy $\pi:\mathcal{S}\rightarrow\mathcal{A}$, is computed as follows:  
\begin{equation}\label{eq:V-func}
    V^{(i)}_\pi(s) = \underset{a \in \mathcal{A}}{\operatorname{max}} ~ Q^{(i)}_\pi(s, a)
\end{equation}
using the Bellman equation \cite{bellman1957dynamic}, where $Q^{(i)}_\pi(s, a)$ denotes the state--action value function (Q-function) defined as follows:
\begin{equation}\label{eq:Q-func}
    Q^{(i)}_\pi(s, a) = \sum_{s^\prime \in \mathcal{S}} w_{\rm cpt}(p(s^\prime|s, a))  \Big(r_\mu(s, a, s^\prime) + \gamma V^{(i-1)}_\pi(s^\prime)\Big)
\end{equation}
with $r_\mu(s, a, s^\prime) = r(s, a, s^\prime) - \mu v^+_{\rm cpt}(f_c(a))$ as a \emph{net} reward. Here, $\mu$ is given from the outer loop.
Accordingly, the scheduling policy for state $s\in\mathcal{S}$ is improved by
\begin{equation}\label{eq:opt-policy}
    \pi(s) \in \underset{a \in \mathcal{A}}{\operatorname{arg\,max}} ~ Q^{(i)}_\pi(s, a).
\end{equation}

The value iteration stops at the $i$-th step, and $\pi(s)$ from \eqref{eq:opt-policy} converges to the optimal policy, i.e., $\pi^*(s)$, given $\mu$, once the following criterion is met \cite{puterman2014markov}:
\begin{equation}\label{eq:sp-func}
    \operatorname{sp}\!\big(V^{(i)}_\pi - V^{(i-1)}_\pi\big) < \epsilon_\pi
\end{equation}
where $\epsilon_\pi>0$ denotes the desired convergence sensitivity. Also, $\operatorname{sp}:\mathbb{R}_0^+\rightarrow\mathbb{R}_0^+$ is the span function, defined as
\begin{equation}
     \operatorname{sp}\!\big(V^{(i)}_\pi\big) = \underset{s \in \mathcal{S}}{\operatorname{max}} ~ V^{(i)}_\pi(s) - \underset{s \in \mathcal{S}}{\operatorname{min}} ~ V^{(i)}_\pi(s)
\end{equation}
based on the span seminorm \cite[Section~6.6.1]{puterman2014markov}. Given Corollary~\ref{cor:cor_2} and the fact that every optimal class of policies has an aperiodic transition matrix under the modeled CMDP, the criterion in \eqref{eq:sp-func} is guaranteed to be satisfied after a finite number of iterations \cite[Theorem~8.5.4]{puterman2014markov}.

\subsubsection{Computing $\mu^*$}
As long as the cost constraint of $\mathcal{P}$ is not satisfied, i.e., while the outer loop is running, $\mathcal{L}(\pi^*, \mu)$ from \eqref{eq:lag} remains a non-increasing function of $\mu$ for some $\pi^*$. Consequently, the bisection method searches for the minimum Lagrange multiplier that satisfies the cost constraint. 

The search starts from an initial interval $[\mu_{-}^{(0)}, \mu_{+}^{(0)}]$ complying with $\mathcal{L}(\pi^*, \mu_{-}^{(0)})\mathcal{L}(\pi^*, \mu_{+}^{(0)}) < 0$ given $\pi^*$ from the inner loop. In each subsequent step, the Lagrange multiplier is updated to the midpoint of the current interval. In this context, we have $\mu^{(l)} = \frac{\mu_{-}^{(l)} + \mu_{+}^{(l)}}{2}$ for the $l$-th step, where $l\in\mathbb{N}$. The interval is updated after computing $\pi^*$ based on the new multiplier and evaluating whether the cost constraint is met.
The bisection search terminates when the stopping criterion, as defined below, is met:
\begin{equation}
    \big\lvert\mu_{+}^{(l)}-\mu_{-}^{(l)}\big\rvert < \epsilon_\mu
\end{equation}
where $\epsilon_\mu>0$ indicates the acceptable tolerance.
It can be shown that $\mathcal{L}(\pi^*, \mu)$ is a Lipschitz continuous function of $\mu$, with the Lipschitz constant given by
\begin{equation*}
    \left\lvert\, C_{\rm max} - \mathbb{E}_{\pi}\!\left[\,\sum_{t = 0}^{\infty} \gamma^t v^+_{\rm cpt}\big(f_c(\mathbf{1}_N^T \mathbf{A}(t) \mathbf{1}_M)\big) \,\Big|\,\pi \right]\right\rvert.
\end{equation*}
Thus, the bisection search converges to the optimal Lagrange multiplier within a finite number of steps \cite{Wood2009}. 

Once both loops terminate, the resulting scheduling policy, i.e., $\pi^*$, given by the Lagrange multiplier, i.e., $\mu^*$, is \emph{deterministic} if the expected discounted cumulative query cost exactly equals $C_{\rm max}$. Otherwise, the policy is \emph{randomized stationary}, derived by probabilistically mixing two deterministic policies, $\pi^-$ and $\pi^+$, with a computable probability of $\eta\in[0, 1]$, where
\begin{equation*}
    \pi^- = \underset{\mu \rightarrow \mu_{-}^{(l)}}{\lim} \pi, ~~ \pi^+ = \underset{\mu \rightarrow \mu_{+}^{(l)}}{\lim} \pi,
\end{equation*}
in the $l$-th step after which the bisection search terminates. Hence, for the state $s\in\mathcal{S}$, we have
\begin{equation}
    \pi^*(s) \leftarrow \eta \pi^-(s) + (1-\eta)\pi^+(s).
\end{equation}
This means that the scheduling policy is randomly selected such that $\pi^*(s) = \pi^-(s)$ with probability $\eta$, and $\pi^*(s) = \pi^+(s)$, with probability $1-\eta$, $\forall s\in\mathcal{S}$.
% Brent's Method
% \begin{equation}
%     \mu = \mu^{(l-1)} - 
%     \dfrac{
%     (\mu^{(l-1)} - \mu^{(l-2)})^2\big[\mathcal{L}(\pi, \mu^{(l-1)}) - \mathcal{L}(\pi, \mu^{(l-2)})\big] - (\mu^{(l-1)} - \mu^{(l-3)})^2\big[\mathcal{L}(\pi, \mu^{(l-1)}) - \mathcal{L}(\pi, \mu^{(l-3)})\big]
%     }{
%     2((\mu^{(l-1)} - \mu^{(l-2)})\big[\mathcal{L}(\pi, \mu^{(l-1)}) - \mathcal{L}(\pi, \mu^{(l-2)})\big] - (\mu^{(l-1)} - \mu^{(l-3)})\big[\mathcal{L}(\pi, \mu^{(l-1)}) - \mathcal{L}(\pi, \mu^{(l-3)})\big])}
% \end{equation}

\subsubsection{Complexity analysis}
Performing value iteration for a number of iterations that scale \emph{polynomially} with $\lvert \mathcal{S}\rvert$, $|\mathcal{A}|$, and $\frac{1}{1-\gamma}\log(\frac{1}{1-\gamma})$ ensures the computation of the optimal scheduling policy, given a fixed $\gamma$ and the Lagrange multiplier \cite{littman2013complexity}. Additionally, the bisection search requires $\lceil \log_2(\frac{\mu_{+}^{(0)}}{\epsilon_\mu})\rceil$ steps to find the optimal Lagrange multiplier with a tolerance of $\epsilon_\mu$ under the derived policy, where $\mu_{+}^{(0)}$ is the upper bound of the initial interval. Hence, the overall time complexity of Algorithm~\ref{alg_1}, accounting for both the inner and outer loops, is given by
\begin{equation*}
\mathcal{O}\!\left(\frac{\lvert \mathcal{S}\rvert|\mathcal{A}|}{1-\gamma}\log\!\left(\frac{1}{1-\gamma}\right) \!\log\!\left(\frac{\mu_{+}^{(0)}}{\epsilon_\mu}\right)\!\right).
\end{equation*}

The algorithm's complexity increases with larger state and action spaces, a wider initial interval for the Lagrange multiplier, and as $\gamma \rightarrow 1$ and $\epsilon_\mu \rightarrow 0$. Therefore, the model-based solution faces the significant drawback of high computational and processing demands when scaling to larger models.

\section{Model-Free Solutions for Effect-Aware Scheduling}\label{sec5}
In the following, we introduce \emph{model-free} solutions for deriving effect-aware scheduling policies, tackling the scalability challenges associated with the model-based approach (see Section~\ref{sec4}). 

\subsection{Model-Free Environment}
To address the limitations of relying on a specific model, we consider an \emph{environment} that encompasses the source, sensing agents, error-prone channels for transporting noisy updates to the hub, and the evolving knowledge base. In this model-free approach, the hub interacts with the environment by performing actions, making observations, and receiving rewards, all without prior knowledge of the source's dynamics, the potential usefulness of the updates to be queried, or the status of the update channels.

Within this framework, the hub takes an action $a(t)\in\mathcal{A}$ while the environment is in state $s(t)\in\mathcal{S}$ during the $t$-th slot. 
The action is passed to the environment, and after consulting the updated knowledge base, the hub observes the new state of the environment, $s(t+1)$. Based on this transition, the immediate net reward $r(t)$ resulting from the chosen action is computed. The recurring interaction between the hub and the environment over $T_{\rm e}$ steps is represented as a sequence $\langle s(t), a(t), r(t), s(t+1) \rangle_{t=0}^{T_{\rm e}}$. Specifically, the reward $r(t)$ at the $t$-th slot is given by
\begin{align}\label{eq:reward_learn}
    r(t) &= r_\mu(s(t), a(t), s(t+1))\nonumber \\
    &=v_{\rm cpt}(\operatorname{GoE}(t+1)) - \mu v^+_{\rm cpt}(f_c(a(t)))
\end{align}
where $\mu$ is the Lagrange multiplier (see Section~\ref{sec4}).

This formalism defines the state space $\mathcal{S}$, the action space $\mathcal{A}$, and the net rewards of the modeled CMDP in Section~\ref{sec3:partB}. Leveraging the model-free approach, we develop a \emph{learning-based} iterative algorithm to derive policies for solving the dual scheduling problem, i.e., $\widehat{\mathcal{P}}$, from Section~\ref{sec4:partA}.

\subsection{Learning-Based Iterative Algorithm}
We adopt a similar iterative process to that outlined in Algorithm~\ref{alg_1}. The main distinction here is that we use model-free, learning-based solutions to find the class of effect-aware scheduling policies, $\pi$, within the inner loop.

\subsubsection{Computing $\pi$}
To derive scheduling policies, we adapt Deep Q-Network (DQN) \cite{DQN}, Advantage Actor--Critic (A2C) \cite{A2C}, and Proximal Policy Optimization (PPO) \cite{PPO}. 
We begin with the off-policy DQN algorithm, which approximates the optimal state--action value function $Q^*_\pi(s(t), a(t))$ using a deep neural network. The Q-values are updated by minimizing the temporal-difference loss between predicted and target Q-values, as defined by the Bellman equation. In our formulation, the Q-function and value function are computed as in \eqref{eq:V-func}--\eqref{eq:Q-func}, based on the net reward defined in \eqref{eq:reward_learn}. To improve training stability, DQN incorporates two key techniques: \emph{experience replay}, which mitigates sample correlation by learning from a buffer of past transitions, and a \emph{target network}, which is updated at a slower rate to stabilize Q-value estimates.

In contrast, A2C is an on-policy, synchronous actor--critic method in which the actor learns a stochastic policy $\pi(s(t))$, and the critic estimates the value function $V_\pi(s(t))$ for $s(t)\in\mathcal{S}$. To reduce the variance of policy gradient updates, it employs the advantage function $A_\pi(s(t), a(t))=Q_\pi(s(t), a(t))-V_\pi(s(t))$, which measures how much better action $a(t)$ is relative to the expected value of state $s(t)$ at the $t$-th time slot. A2C performs multiple parallel rollouts across environments, aggregates gradients, and synchronously updates both the actor and critic networks. This design fosters more stable training while preserving the sample efficiency of actor--critic methods.

Finally, PPO, another on-policy algorithm, also follows an actor--critic architecture while enhancing training stability through constrained policy updates. It directly optimizes a stochastic policy by maximizing a \emph{clipped surrogate objective}, which penalizes large deviations from the previous policy. This clipping mechanism limits abrupt changes in the action probability ratio between successive policies, thereby mitigating the risk of overly aggressive updates that can destabilize learning. By enforcing conservative yet effective policy adjustments, PPO improves upon earlier policy gradient methods and can offer more reliable convergence across a wide range of tasks.

These algorithms exhibit complementary strengths and trade-offs: A2C and PPO, due to their on-policy design, provide stable policy updates in dynamic environments, making them effective in scenarios where adaptability is crucial. However, this comes at the cost of lower sample efficiency, as they require extensive interaction with the environment, potentially slowing convergence. In contrast, DQN benefits from off-policy learning and experience replay, enabling higher sample efficiency and making it particularly well-suited for discrete action spaces. Nevertheless, DQN may face challenges in high-dimensional or partially observable environments, where its value-based estimation can lead to instability and limited generalization across diverse state–action pairs.\footnote{More recent or specialized algorithms, such as Soft Actor--Critic (SAC) and Twin Delayed Deep Deterministic Policy Gradient (TD3), are excluded from consideration, as they are tailored for continuous control tasks and introduce additional complexity without offering clear benefits for the discrete scheduling objective addressed in this work.}

\subsubsection{Computing $\mu^*$}
We employ the same approach as in Algorithm~\ref{alg_1} to determine the optimal Lagrange multiplier for a given policy from the inner loop. In this context, the bisection search method is employed in the outer loop to iteratively and gradually identify the minimum Lagrange multiplier that meets the cost constraint of the scheduling problem $\mathcal{P}$, as formulated in \eqref{eq:opt-I}.

\section{Simulation Results}
In this section, we evaluate the performance of the proposed model-based and model-free solutions outlined in Sections~\ref{sec4} and \ref{sec5}, respectively, within the context of effect-aware query scheduling. To assess their effectiveness, we compare these solutions against several well-established benchmark approaches.

\subsection{Setup and Assumptions}
We consider a system with $N = 4$ SAs observing a source characterized by $M = 2$ attributes and $K = 4$ AAs performing actions over $T = 1,000$ slots. The usefulness of an update on the $m$-th attribute is mapped to a value within the range $[0, 1]$, determined by applying 
\begin{equation}
    g_m(t; y_m(t)) = \operatorname{min}\left\{ 1, \frac{y_m^{\alpha_m-1}(t) (1 - y_m(t))^{\beta_m-1}}{\operatorname{B}(\alpha_m, \beta_m)} \right\}
\end{equation}
at the $t$-th slot, where $\alpha_m, \beta_m>0, \forall m$, are shape parameters, and $\operatorname{B}(\cdot, \cdot)$ is Beta function. Besides, we consider $\operatorname{GoE}_m(t) = \frac{u_m(t)}{\Delta_m(t)}, \forall m$, and the CPT-based value function for an arbitrary $x\in\mathbb{R}$ is given by \cite{tversky1992advances}
\begin{align}
    v_{\rm cpt}(x) = 
    \begin{cases}
         v^+_{\rm cpt}(x) = (x - x_{\rm ref})^{\alpha_{\rm cpt}}, ~x \geq x_{\rm ref};\\
         v^-_{\rm cpt}(x) = -\lambda_{\rm cpt} (x_{\rm ref} - x)^{\beta_{\rm cpt}}, ~x < x_{\rm ref},
    \end{cases}
\end{align}
with $\alpha_{\rm cpt}=\beta_{\rm cpt}=0.5$, and $\lambda_{\rm cpt}=2$ being shape parameters under the given reference point $x_{\rm ref}$. For simplicity and to facilitate comparison with non-probabilistic scheduling methods, we assume $w_{\rm cpt}(x)=x, \forall x$. 

Finally, the cost constraint is enforced by introducing a \emph{cost flexibility index} $C_{\rm flex}$ multiplied by the discounted cumulative cost incurred from querying across \emph{all} slots. Thus, we have
\begin{equation}\label{eq:c-flex}
    C_{\rm max} = C_{\rm flex} \left[\sum_{t = 0}^{\infty} \gamma^t v^+_{\rm cpt}\big(f_c(1)\big)\right] = C_{\rm flex} \frac{ v^+_{\rm cpt}\big(f_c(1)\big)}{1 - \gamma}
\end{equation}
where $f_c(1)$ indicates the fixed cost per query. Unless stated otherwise, the default simulation parameter values are outlined in Table~\ref{tab:params}.
\noindent
\begin{table}
    \centering
    \caption{Parameters for simulation results.}\label{tab:params}
    \begin{tabular}{|l|c|c||l|c|c|}
    \hline
         {\!\!\footnotesize \textbf{Parameter}\!\!}&
         {\!\!\footnotesize \textbf{Symbol}\!\!}&
         {\!\!\footnotesize \textbf{Value}\!\!}&
         {\!\!\footnotesize \textbf{Parameter}\!\!}&
         {\!\!\footnotesize \textbf{Symbol}\!\!}&
         {\!\!\!\footnotesize \textbf{Value}\!\!\!} \\
         \hline
         \hline
         {\!\!\footnotesize Number of slots\!\!}&
         {\!\!\footnotesize $T$\!\!}& 
         {\!\!\footnotesize $10^3$\!\!}&
         \multirow{3}{*}{\makecell[l]{\!\!\footnotesize Parameters\!\! \\\!\!\footnotesize shaping $v_{\rm cpt}(\cdot)$\!\!\!}}
         &{\!\!\footnotesize $\alpha_{\rm cpt}$\!\!}& \multirow{2}{*}{\!\!\footnotesize $0.5$\!\!}\\
         \cline{1-3}\cline{5-5}
         {\!\!\footnotesize Number of SAs\!\!}&
         {\!\!\footnotesize $N$\!\!}& 
         {\!\!\footnotesize $4$\!\!}&
         &{\!\!\footnotesize $\beta_{\rm cpt}$\!\!}&\\
         \cline{1-3}\cline{5-6}
         {\!\!\footnotesize Number of attributes\!\!\!}&
         {\!\!\footnotesize $M$\!\!}&
         {\!\!\footnotesize $2$\!\!}&  
         &{\!\!\footnotesize $\lambda_{\rm cpt}$\!\!}& 
         {\!\!\footnotesize $2$\!\!}\\
         \hline
         {\!\!\footnotesize Number of AAs\!\!}&
         {\!\!\footnotesize $K$\!\!}&
         {\!\!\footnotesize $4$\!\!}&  
         {\!\!\footnotesize Refs. point\!\!}&
         {\!\!\footnotesize $\operatorname{GoE}_{\rm ref}$\!\!}&{\!\!\footnotesize $0.2$\!\!}\\
         \hline
         \makecell[l]{\!\!\footnotesize Corr. observation\!\! \\\!\!\footnotesize probability\!\!}&
         \makecell{\!\!\footnotesize $p_{{\rm{o}},nm},$\!\!\\ \!\!\footnotesize $\forall n, m$\!\!}&
         {\!\!\footnotesize $0.8$\!\!}&
         {\!\!\footnotesize Cost per query\!\!}&
         {\!\!\footnotesize $f_c(1)$\!\!}& 
         {\!\!\footnotesize $0.5$\!\!}\\
         \hline
         {\!\!\footnotesize Required attributes\!\!}&
         {\!\!\footnotesize $\lvert\mathcal{M}_k\rvert, \forall k$\!\!}&
         {\!\!\footnotesize $2$\!\!}
         & 
         {\!\!\footnotesize Cost flex. index\!\!\!}&
         {\!\!\footnotesize $C_{\rm flex}$\!\!}& {\!\!\footnotesize $0.75$\!\!}\\
         \hline
         {\!\!\footnotesize Erasure probability\!\!}&
         {\!\!\footnotesize $p_{{\rm{e}}, n}, \forall n$\!\!}& 
         {\!\!\footnotesize $0.2$\!\!}&  
         {\!\!\footnotesize Maximum AoI\!\!}&
         {\!\!\footnotesize $\Delta_{\rm max}$\!\!}&
         {\!\!\footnotesize $4$\!\!}
         \\
         \hline
         {\!\!\footnotesize Discount factor\!\!}&
         {\!\!\footnotesize $\gamma$\!\!}& 
         {\!\!\footnotesize $0.9$\!\!}&
         {\!\!\footnotesize Converg. sens.\!\!}&
         {\!\!\footnotesize $\epsilon_\pi$\!\!}
         &\multirow{2}{*}{\!\!\footnotesize\!$10^{-6}$\!\!\!}\\
         \cline{1-5}
         \multirow{2}{*}{\makecell[l]{\!\!\footnotesize Shape parameters\!\! \\\!\!\footnotesize for $g_{m}(\cdot;\cdot), \forall m$\!\!}}&
         {\!\!\footnotesize $\{\alpha_1, \alpha_2\}$\!\!}&
         {\!\!\footnotesize\!$\{0.5, 2\}$\!\!\!}& 
         {\!\!\footnotesize Tolerance sens.\!\!}&  
         {\!\!\footnotesize $\epsilon_\mu$\!\!}&\\
         \cline{2-6}
         &{\!\!\footnotesize $\{\beta_1, \beta_2\}$\!\!}&
         {\!\!\footnotesize\!$\{0.5, 5\}$\!\!\!}&
         {\!\!\footnotesize Mixing factor\!\!}&  
         {\!\!\footnotesize $\eta$\!\!}&{\!\!\footnotesize $0.5$\!\!}\\
         \hline
         {\!\!\footnotesize Size of usefulness\!\!}&
         {\!\!\footnotesize $\lvert\mathcal{U}\rvert$\!\!}& {\!\!\footnotesize $4$\!\!}
         & --
         & --
         & -- \\
         \hline
    \end{tabular}
\end{table}

To implement the DRL algorithms, namely A2C, PPO, and DQN, we mostly follow the default hyperparameter settings as considered in their respective original papers, i.e., \cite{A2C}, \cite{PPO}, and \cite{DQN}. However, we adjust the discount factor according to Table~\ref{tab:params}. We employ the Adaptive Moment Estimation (Adam) optimizer for both DQN and PPO, whereas the Root Mean Square Propagation (RMSprop) optimizer is used for A2C. For each DRL algorithm, a model is trained over $100$ episodes of interaction. Each episode comprises $10,000$ time slots, generating a total of $T_{\rm e} = 10^6$ environment steps. The key hyperparameters for these algorithms, along with their initial configurations, are summarized in Table~\ref{tab:RL-hyparams}.
\noindent
\begin{table}
    \centering
    \caption{Hyperparameters for DRL algorithms.}\label{tab:RL-hyparams}
    \begin{tabular}{|l|c||c||c|}
    \hline
         {\!\!\footnotesize \textbf{Hyperparameter}}&
         {\footnotesize \textbf{DQN}}&
         {\footnotesize \textbf{A2C}}&
         {\footnotesize \textbf{PPO}} \\
    \hline
    \hline
         {\!\!\footnotesize Neural network model}&
         \multicolumn{3}{c|}{\footnotesize Multi-layer perceptron (MLP)}\\
    \hline
         {\!\!\footnotesize Hidden layers activation}&
         \multicolumn{3}{c|}{\footnotesize Rectified linear unit (ReLU)}\\
    \hline
         {\!\!\footnotesize Policy net. output activation\!}&
         -- &
         \multicolumn{2}{c|}{\footnotesize Softmax}\\
    \hline
         {\!\!\footnotesize Optimizer class}&
         {\footnotesize Adam}& 
         {\footnotesize RMSprop}&
         {\footnotesize Adam}\\
    \hline
         {\!\!\footnotesize Learning rate}&
         {\footnotesize $10^{-4}$}& 
         {\footnotesize $3\!\times\! 10^{-4}$}&
         {\footnotesize $3\!\times\! 10^{-4}$}\\
    \hline
        {\!\!\footnotesize Environment steps}&
         \multicolumn{3}{c|}{\footnotesize $10^2\!\times\!10^4$}\\
    \hline
         {\!\!\footnotesize Replay buffer size}&
         {\footnotesize $10^{6}$}& 
         -- &
         -- \\
    \hline
         {\!\!\footnotesize Rollout buffer size}&
         --& 
         {\footnotesize $20\!\times\!8$}&
         {\footnotesize $1\!\times\!2048$}\\
    \hline
         {\!\!\footnotesize Mini-batch size}&
         {\footnotesize $32$}& 
         -- &
         {\footnotesize $64$}\\
    \hline
         {\!\!\footnotesize Number of epochs}&
         -- & 
         -- &
         {\footnotesize $10$}\\
    \hline
         {\!\!\footnotesize Bias-variance trade-off}&
         -- & 
         {\footnotesize $1$}&
         {\footnotesize $0.95$}\\
    \hline
    \end{tabular}
\end{table}

\subsection{Results and Discussion}
To conduct the comparison, we use four benchmark scheduling approaches: (1) \emph{weighted round-robin} (WRR), (2) \emph{lowest-weighted-grade-first} (LWGF), and scheduling strategies based on either (3) a \emph{uniform} distribution or (4) a \emph{Markovian} process. In the LWGF approach, the attribute with the lowest weighted GoE from the previous time slot is prioritized for scheduling in the current slot. Both WRR and LWGF assign weights to attributes based on their importance to the AAs. The transition probability matrix for the Markovian process is designed to ensure that the cost constraint is satisfied over the long term.

Fig.~\ref{Fig:cdf_goe} depicts the cumulative distribution function (CDF) of the long-term CPT-based total GoE for different scheduling methods over $1,000$ times slots, comparing the model-based effect-aware scheduling as a reference against (a) benchmark scheduling approaches and (b) model-free effect-aware ones. According to Fig.~\ref{Fig:cdf_goe_nonlearn}, querying under either a uniform or a Markovian process fails to achieve high performance, primarily because these approaches neglect the effectiveness of updates. While WRR achieves higher performance, it faces scalability challenges as the number of attributes increases. Notably, model-based effect-aware scheduling achieves performance comparable to LWGF, with only a $1.37\%$ difference, the latter delivering the highest effectiveness overall. 
\begin{figure}[t!]
    \centering
    \subfloat[]{
    \includegraphics[width=0.4\textwidth]{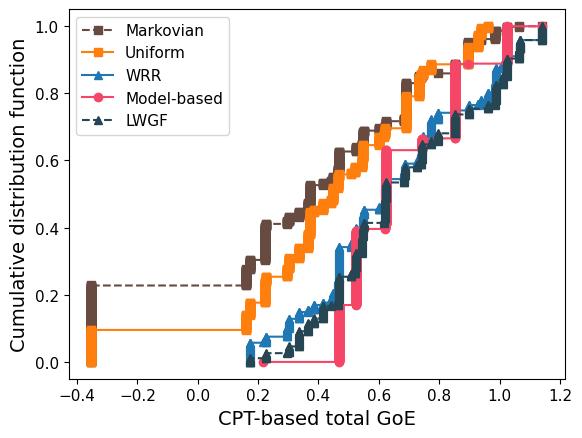}
    \label{Fig:cdf_goe_nonlearn}
    }
    \hfil
    \subfloat[]{
    \includegraphics[width=0.4\textwidth]{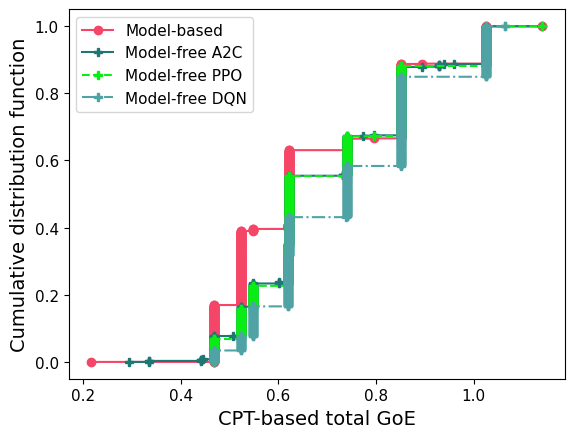}
    \label{Fig:cdf_goe_learn}
    }
    \caption{CDF of the long-term CPT-based total GoE over $1,000$ times slots.}\label{Fig:cdf_goe}
\end{figure}

However, looking at the bar chart in Fig.~\ref{Fig:query_stats}, it becomes clear that LWGF scheduling results in a significantly higher frequency of hub queries to the SAs compared to model-based scheduling. Furthermore, Fig.~\ref{Fig:cdf_goe_learn} shows that employing model-free effect-aware scheduling can improve long-term effectiveness by up to $10.57\%$ compared to the model-based approach, with $14.11\%$ increase in the number of queries sent. In this scenario, DQN outperforms both PPO and A2C, delivering $5.37\%$ and $5.77\%$ higher effectiveness, respectively.

The bar chart in Fig.~\ref{Fig:query_stats} illustrates the percentage of queries sent within $1,000$ slots, along with the percentages of successful update communication. Despite achieving similar success rates in communicating updates across all approaches, model-based, model-free PPO, and model-free A2C stand out for sending (approximately) the fewest queries over the evaluated period. Specifically, these approaches send around $20\%$ fewer queries compared to LWGF, WRR, and uniform scheduling, $14\%$ fewer than model-free DQN, and $9\%$ fewer than the Markovian process. This highlights that effect-aware scheduling, on average, strikes the best \emph{balance} between effectiveness and query efficiency, enabling higher efficiency in generating and transmitting updates while maintaining a high level of effectiveness.
\begin{figure}[t!]
    \centering
    \includegraphics[width=0.4\textwidth]{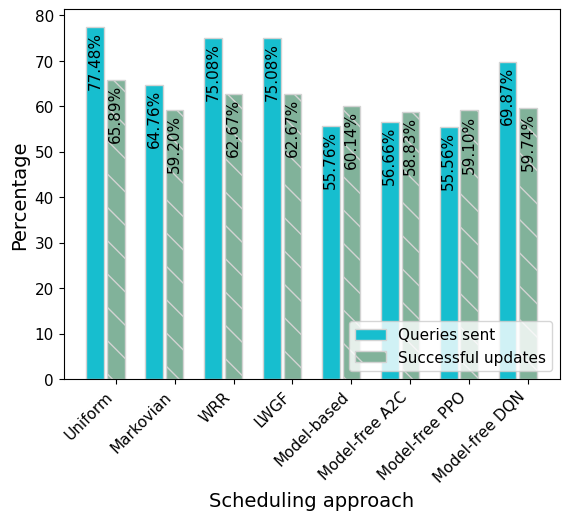}
    \caption{Percentages of queries sent over $1,000$ time slots, alongside the statistics of successful and failed update communications.}
    \label{Fig:query_stats}
\end{figure}

Expanding on the time-variant effectiveness provided by different approaches, Fig.~\ref{Fig:goe_vs_time} depicts the instantaneous CPT-based total GoE at each slot over the interval $[0, 50]$. It highlights that, unlike other approaches that show significant fluctuations in performance, both model-based and model-free effect-aware scheduling demonstrate consistent and stable performance, ensuring a minimum level of effectiveness. Specifically, in Fig.~\ref{Fig:goe_vs_time_learn}, effect-aware scheduling consistently maintains a CPT-based total GoE of $0.29$ or higher throughout the observed period after passing the initial state, i.e., $t=0$. In contrast, using LWGF, WRR, uniform, and Markovian, the CPT-based total GoE could decrease to as low as $0.23$, $0.17$, $-0.35$, and $-0.35$, respectively, as shown in Fig.~\ref{Fig:goe_vs_time_nonlearn}. This confirms that effect-aware scheduling outperforms the benchmark approaches in consistently maintaining reliable and stable performance.
\begin{figure}[t!]
    \centering
    \subfloat[]{
    \includegraphics[width=0.4\textwidth]{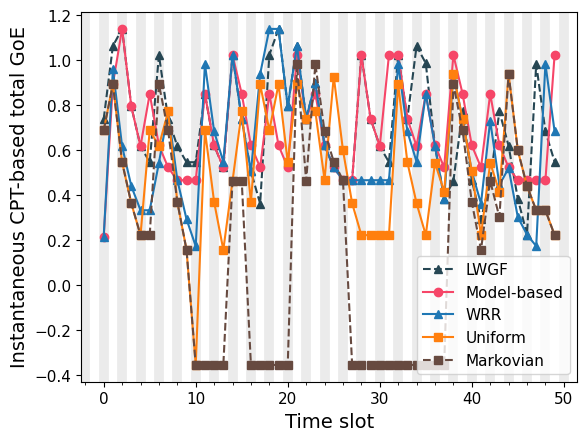}
    \label{Fig:goe_vs_time_nonlearn}
    }
    \hfil
    \subfloat[]{
    \includegraphics[width=0.4\textwidth]{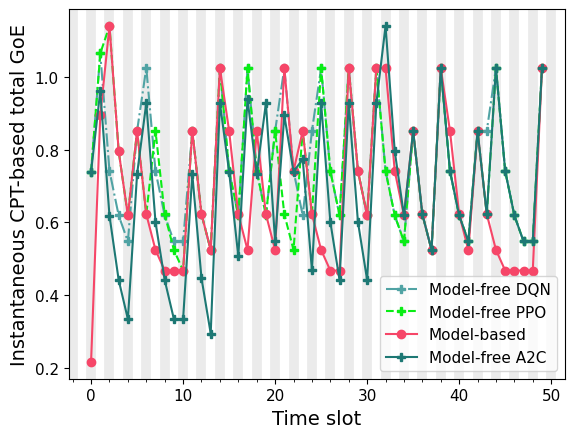}
    \label{Fig:goe_vs_time_learn}
    }
    \caption{Evolution of the instantaneous CPT-based total GoE over time within the period of $[0, 50]$.}\label{Fig:goe_vs_time}
\end{figure}

To elaborate on the computational complexity of the considered model-free DRL algorithms, Fig.~\ref{Fig:convergence} illustrates their convergence behavior by representing the number of episodes required to attain stable performance.
In this figure, the episodic rewards for each algorithm are smoothed using a $20$-episode moving average to better highlight convergence trends. Under the default settings, model-free PPO demonstrates the fastest convergence, stabilizing approximately $58.33\%$ and $75.66\%$ faster than model-free DQN and A2C, respectively. Notably, DQN and A2C exhibit intermittent performance drops, which can be attributed to their exploration strategies. Nevertheless, these algorithms eventually converge to stable policies as the exploration--exploitation balance is progressively optimized. By applying the dual-loop iterative procedure described in Algorithm~\ref{alg_1}, obtaining the optimal policies along with the associated Lagrange multiplier requires, on average, $23$ bisection search steps in the outer loop across all model-free algorithms. Specifically, deriving the final optimal policies for query scheduling involves a total of $1,155$ ($21 \times 55$), $3,036$ ($23 \times 132$), and $5,876$ ($26 \times 226$) episodes for PPO, DQN, and A2C, respectively. In contrast, the model-based solution converges in just $378$ ($18 \times 21$) iterations, highlighting its significantly lower computational cost. This translates into around $67\%$--$93\%$ reduction in the number of training episodes, underscoring the faster convergence and reduced complexity of the model-based alternative.
\begin{figure}[t!]
    \centering
    \includegraphics[width=0.4\textwidth]{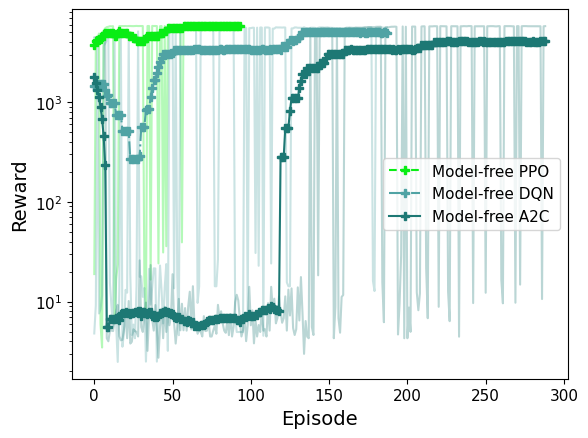}
    \caption{Convergence performance of different model-free DRL algorithms in terms of episodic rewards and $20$-episode moving average rewards.}
    \label{Fig:convergence}
\end{figure}

To study the interplay between the average CPT-based total GoE and its corresponding reference point, i.e., $\operatorname{GoE}_{\rm ref}$, we plot Fig.~\ref{Fig:goe_vs_ref}. This graph depicts the effect of increasing the reference point on the effectiveness and performance of different scheduling approaches. As $\operatorname{GoE}_{\rm ref}$ increases, the effectiveness delivered by all approaches decreases. This decrease can be attributed to the activation of the loss component in the CPT-based value function at higher levels of total GoE. However, as the reference point increases, the performance gap between effect-aware scheduling and the other approaches widens, reaching its peak at $\operatorname{GoE}_{\rm ref}=0.5$. Beyond this critical point, the performance of both model-based and model-free approaches starts to degrade. Notably, after $\operatorname{GoE}_{\rm ref}=0.95$, the benchmark approaches outperform all effect-aware ones. 

The reversal likely arises because increasing the reference point compresses the net reward differences across the $\lvert \mathcal{S}\rvert=256$ states, which in turn distorts the optimal value of the Lagrange multiplier. Consequently, this leads to misalignment in the scheduling policies, both model-based and model-free, with respect to the intended effectiveness objectives. This degradation is a direct outcome of the CPT-based evaluation, where even slight changes in perceived gains or losses can substantially affect action preferences. The extent of this distortion depends on the specific effect-aware solution, illustrating a key insight of integrating CPT: higher expectations---reflected by an elevated reference point---can paradoxically lead to lower effectiveness due to overly risk-averse or misdirected decisions.
\begin{figure}[t!]
    \centering
    \includegraphics[width=0.4\textwidth]{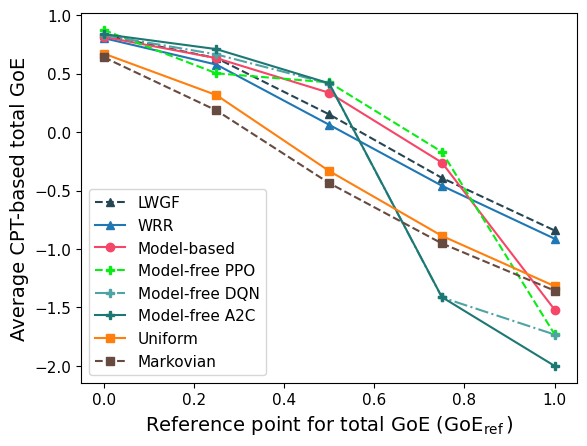}
    \caption{Relationship between the average CPT-based total GoE and its modeling reference point, along with the average percentage of queries sent across different approaches over $1,000$ time slots.}
    \label{Fig:goe_vs_ref}
\end{figure}

Moreover, Fig.~\ref{Fig:goe_vs_cost_flex} demonstrates how lowering the cost constraint affects the average CPT-based total GoE. This effect is analyzed using the cost flexibility index, denoted as $C_{\rm flex}$ in \eqref{eq:c-flex}, where a lower index corresponds to a stricter constraint. The figure highlights the robust performance of effect-aware scheduling under strict cost constraints, consistently outperforming the benchmark approaches. Decreasing the maximum induced cost, particularly below $C_{\rm flex}=0.52$, amplifies the effectiveness performance gap between effect-aware scheduling and the other approaches. For instance, at $C_{\rm flex}=0.286$, the average CPT-based total GoE achieved by model-based and model-free effect-aware scheduling is $4.22$ and $3.25$ times higher than that achieved using LWGF, respectively. This indicates that in scenarios where resources for computation, update generation, or communication, such as energy, are severely limited, effect-aware scheduling delivers significantly superior performance.
\begin{figure}[t!]
    \centering
    \includegraphics[width=0.4\textwidth]{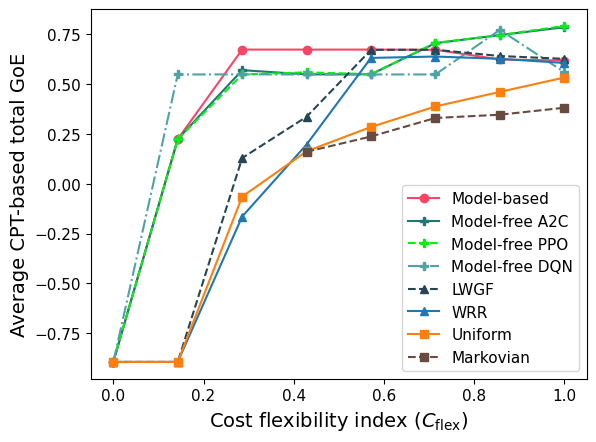}
    \caption{Interplay between the average CPT-based total GoE and the cost flexibility index across $1,000$ slots.}
    \label{Fig:goe_vs_cost_flex}
\end{figure}

To study the scalability of different effect-aware scheduling approaches, Fig.~\ref{Fig:goe_vs_no_attribute} illustrates the impact of the number of attributes $M$ on the average CPT-based total GoE over $1,000$ time slots. In this figure, we assume $\lvert\mathcal{M}_k\rvert = M, \forall k$. The graph shows that model-based scheduling faces significant scalability challenges as the number of attributes increases. The reason behind this could be the rapid growth in the number of defined CMDP states, which grows \emph{exponentially} as $16^M$ with the rise of $M$ (see Section~\ref{sec3:partB}). The increase in the number of states may lead to transition probabilities converging to similar values, causing the value iteration process to generate a fixed class of policies. In contrast, model-free approaches demonstrate high scalability at the expense of requiring a larger number of environment steps to achieve similar performance.
\begin{figure}[t!]
    \centering
    \includegraphics[width=0.4\textwidth]{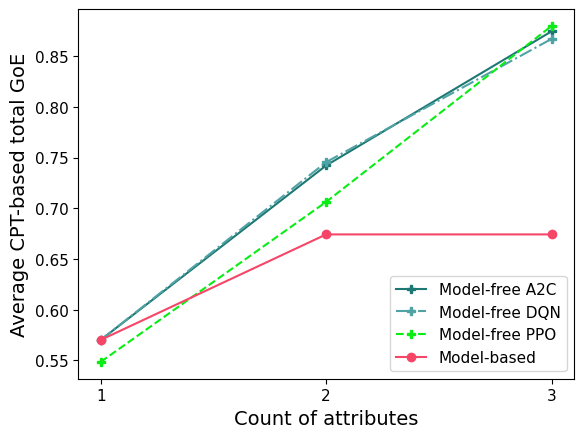}
    \caption{Effect of the number (count) of attributes on the average CPT-based total GoE over $1,000$ slots.}
    \label{Fig:goe_vs_no_attribute}
\end{figure}

Regarding the scalability of model-free effect-aware scheduling with a large number of attributes, as depicted in Fig.~\ref{Fig:goe_vs_no_attribute}, we analyze how increasing the query limit affects effectiveness. Under the constraint $\mathbf{1}_N^T \mathbf{A}(t) \mathbf{1}_M \leq 1$ in \eqref{eq:opt-I}, the hub is essentially restricted to sending at most one query during each sampling interval $t$. Fig.~\ref{Fig:goe_vs_query_limit} illustrates the average CPT-based total GoE over $1,000$ time slots as a function of the query limit, using model-free approaches. Their performance is compared against that of LWGF and uniform scheduling.\footnote{Setting the query limit to more than one increases the complexity of the model-free solution and reduces its appeal for exploration in this context, as this requires corresponding modifications to the CMDP model in Section~\ref{sec3:partB}.} To plot the corresponding curves, we assume $\lvert\mathcal{M}_k\rvert=3, \forall k$ and vary the query limit from $1$ to $3$.\footnote{Given that update packets are typically small, we assume that multiple attribute updates can be transmitted simultaneously through orthogonal channels or using time-division techniques.} The figure shows that increasing the query limit enhances effectiveness across all approaches. Notably, model-free effect-aware scheduling consistently outperforms LWGF on average, particularly when $\mathbf{1}_N^T \mathbf{A}(t) \mathbf{1}_M \leq 2, \forall t$. Both approaches demonstrate a significant performance advantage over uniform scheduling. When $\mathbf{1}_N^T \mathbf{A}(t) \mathbf{1}_M \leq 3$, the hub can query all attributes simultaneously while still adhering to the cost constraint. At this point, the performance of model-free effect-aware scheduling and LWGF converges, becoming nearly identical. 
\begin{figure}[t!]
    \centering
    \includegraphics[width=0.4\textwidth]{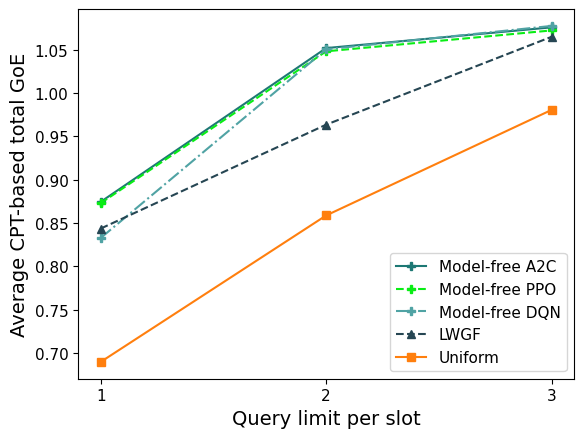}
    \caption{Relationship between the query limit and the average CPT-based total GoE through $1,000$ time slots.}
    \label{Fig:goe_vs_query_limit}
\end{figure}

\section{Conclusion}
In this paper, we studied the query scheduling problem in end-to-end status update systems operating under a pull-based model. We considered multiple SAs observing a source with various attributes and communicating their observations in the form of updates in response to queries from the hub. The hub is responsible for deciding when to query updates and for constructing a knowledge base to enable AAs to achieve their heterogeneous goals. We introduced the GoE metric and formulated the scheduling problem based on this metric, integrating a risk-averse framework. Specifically, we aimed to maximize the expected discounted sum of the CPT-based total GoE while ensuring compliance with a cost constraint. To solve this problem and design effect-aware scheduling policies, we developed model-based and model-free solutions, supported by their respective iterative algorithms. We then evaluated the performance of these solutions in terms of effectiveness and compared them against benchmark methods. Our results demonstrated that the model-based method could improve the efficiency of sent queries by up to $20\%$ while preserving a level of effectiveness comparable to the best benchmark method, namely LWGF. Conversely, model-free solutions showed the potential to enhance long-term effectiveness by up to $9.2\%$ than the other approaches, albeit at the cost of higher computational complexity. This improvement was particularly pronounced under strict cost constraints, where effect-aware scheduling policies can significantly boost effectiveness and outperform others. Our findings highlighted that the model-free solutions excel in scalability while maintaining high levels of effectiveness, making them particularly appealing for large-scale or dynamic decision-making systems.

\balance
\bibliographystyle{IEEEtran}
\bibliography{References.bib}

% Generated by IEEEtran.bst, version: 1.14 (2015/08/26)
\begin{thebibliography}{10}
\providecommand{\url}[1]{#1}
\csname url@samestyle\endcsname
\providecommand{\newblock}{\relax}
\providecommand{\bibinfo}[2]{#2}
\providecommand{\BIBentrySTDinterwordspacing}{\spaceskip=0pt\relax}
\providecommand{\BIBentryALTinterwordstretchfactor}{4}
\providecommand{\BIBentryALTinterwordspacing}{\spaceskip=\fontdimen2\font plus
\BIBentryALTinterwordstretchfactor\fontdimen3\font minus \fontdimen4\font\relax}
\providecommand{\BIBforeignlanguage}[2]{{%
\expandafter\ifx\csname l@#1\endcsname\relax
\typeout{** WARNING: IEEEtran.bst: No hyphenation pattern has been}%
\typeout{** loaded for the language `#1'. Using the pattern for}%
\typeout{** the default language instead.}%
\else
\language=\csname l@#1\endcsname
\fi
#2}}
\providecommand{\BIBdecl}{\relax}
\BIBdecl

\bibitem{kountouris2021semantics}
M.~Kountouris and N.~Pappas, ``Semantics-empowered communication for networked intelligent systems,'' \emph{IEEE Communications Magazine}, vol.~59, no.~6, pp. 96--102, 2021.

\bibitem{popovski2020semantic}
P.~Popovski, O.~Simeone, F.~Boccardi, D.~G{\"u}nd{\"u}z, and O.~Sahin, ``Semantic-effectiveness filtering and control for post-{5G} wireless connectivity,'' \emph{Journal of the Indian Institute of Science}, vol. 100, no.~2, pp. 435--443, 2020.

\bibitem{Strinati2024}
E.~C. Strinati, P.~Di~Lorenzo, V.~Sciancalepore, A.~Aijaz, M.~Kountouris, D.~Gündüz, P.~Popovski, M.~Sana, P.~A. Stavrou, B.~Soret, N.~Cordeschi, S.~Scardapane, M.~Merluzzi, L.~Zanzi, M.~B. Renato, T.~Quek, N.~D. Pietro, O.~Forceville, F.~Costanzo, and P.~Li, ``Goal-oriented and semantic communication in 6g ai-native networks: The 6g-goals approach,'' in \emph{Proceedings of the Joint European Conference on Networks and Communications \& 6G Summit (EuCNC/6G Summit)}, 2024.

\bibitem{pullB}
B.~Yin, S.~Zhang, Y.~Cheng, L.~X. Cai, Z.~Jiang, S.~Zhou, and Z.~Niu, ``Only those requested count: Proactive scheduling policies for minimizing effective age-of-information,'' in \emph{Proceedings of the IEEE International Conference on Computer Communications (INFOCOM)}, 2019, pp. 109--117.

\bibitem{pullC}
F.~Li, Y.~Sang, Z.~Liu, B.~Li, H.~Wu, and B.~Ji, ``Waiting but not aging: Optimizing information freshness under the pull model,'' \emph{IEEE/ACM Transactions on Networking}, vol.~29, no.~1, pp. 465--478, 2021.

\bibitem{pullD}
J.~Holm, A.~E. Kalør, F.~Chiariotti, B.~Soret, S.~K. Jensen, T.~B. Pedersen, and P.~Popovski, ``Freshness on demand: Optimizing age of information for the query process,'' in \emph{Proceedings of the IEEE International Conference on Communications (ICC)}, 2021.

\bibitem{pullDD}
O.~T. Yavascan, E.~T. Ceran, Z.~Cakir, E.~Uysal, and O.~Kaya, ``When to pull data for minimum age penalty,'' in \emph{Proceedings of the IEEE International Symposium on Modeling and Optimization in Mobile, Ad Hoc, and Wireless Networks (WiOpt)}, 2021, pp. 1--8.

\bibitem{pullE}
{F. Chiariotti et al.}, ``Query age of information: Freshness in pull-based communication,'' \emph{IEEE Transactions on Communications}, vol.~70, no.~3, pp. 1606--1622, 2022.

\bibitem{agheli2023effective}
P.~Agheli, N.~Pappas, P.~Popovski, and M.~Kountouris, ``Effective communication: When to pull updates?'' in \emph{Proceedings of the IEEE International Conference on Communications (ICC)}, 2024, pp. 183--188.

\bibitem{kahneman1984choices}
D.~Kahneman and A.~Tversky, ``Choices, values, and frames.'' \emph{American Psychologist}, vol.~39, no.~4, p. 341, 1984.

\bibitem{luce1991rank}
R.~D. Luce and P.~C. Fishburn, ``Rank-and sign-dependent linear utility models for finite first-order gambles,'' \emph{Journal of Risk and Uncertainty}, vol.~4, no.~1, pp. 29--59, 1991.

\bibitem{tversky1992advances}
A.~Tversky and D.~Kahneman, ``Advances in prospect theory: Cumulative representation of uncertainty,'' \emph{Journal of Risk and Uncertainty}, vol.~5, pp. 297--323, 1992.

\bibitem{agheli2024integrated}
P.~Agheli, N.~Pappas, P.~Popovski, and M.~Kountouris, ``Integrated push-and-pull update model for goal-oriented effective communication,'' \emph{IEEE Transactions on Communications (Early Access)}, 2025.

\bibitem{chiariotti2022scheduling}
F.~Chiariotti, A.~E. Kal{\o}r, J.~Holm, B.~Soret, and P.~Popovski, ``Scheduling of sensor transmissions based on value of information for summary statistics,'' \emph{IEEE Networking Letters}, vol.~4, no.~2, pp. 92--96, 2022.

\bibitem{holm2023goal}
J.~Holm, F.~Chiariotti, A.~E. Kal{\o}r, B.~Soret, T.~B. Pedersen, and P.~Popovski, ``Goal-oriented scheduling in sensor networks with application timing awareness,'' \emph{IEEE Transactions on Communications}, vol.~71, no.~8, pp. 4513--4527, 2023.

\bibitem{bui2023scheduling}
V.-P. Bui, S.~R. Pandey, F.~Chiariotti, and P.~Popovski, ``Scheduling policy for value-of-information {(VoI)} in trajectory estimation for digital twins,'' \emph{IEEE Communications Letters}, vol.~27, no.~6, pp. 1654--1658, 2023.

\bibitem{cao2023goal}
J.~Cao, E.~Kurniawan, B.~Amnart, N.~Pappas, S.~Sun, and P.~Popovski, ``Goal-oriented communication, estimation, and control over bidirectional wireless links,'' \emph{IEEE Transactions on Communications}, vol.~73, no.~5, pp. 3031--3045, 2025.

\bibitem{raghuwanshi2024goal}
P.~Raghuwanshi, O.~L.~A. L{\'o}pez, V.~Bhatia, and M.~Latva-aho, ``Goal-oriented sensor reporting scheduling for non-linear dynamic system monitoring,'' \emph{arXiv:2405.20983}, 2024.

\bibitem{AoIA}
A.~M. Bedewy, Y.~Sun, S.~Kompella, and N.~B. Shroff, ``Optimal sampling and scheduling for timely status updates in multi-source networks,'' \emph{IEEE Transactions on Information Theory}, vol.~67, no.~6, pp. 4019--4034, 2021.

\bibitem{AoIB}
M.~Moltafet, M.~Leinonen, and M.~Codreanu, ``Average {AoI} in multi-source systems with source-aware packet management,'' \emph{IEEE Transactions on Communications}, vol.~69, no.~2, pp. 1121--1133, 2020.

\bibitem{AoIC}
A.~E. Kal{\o}r and P.~Popovski, ``Minimizing the age of information from sensors with common observations,'' \emph{IEEE Wireless Communications Letters}, vol.~8, no.~5, pp. 1390--1393, 2019.

\bibitem{AoID}
R.~D. Yates and S.~K. Kaul, ``The age of information: Real-time status updating by multiple sources,'' \emph{IEEE Transactions on Information Theory}, vol.~65, no.~3, pp. 1807--1827, 2018.

\bibitem{AoIE}
I.~Kadota, A.~Sinha, E.~Uysal-Biyikoglu, R.~Singh, and E.~Modiano, ``Scheduling policies for minimizing age of information in broadcast wireless networks,'' \emph{IEEE/ACM Transactions on Networking}, vol.~26, no.~6, pp. 2637--2650, 2018.

\bibitem{stamatakis2023optimizing}
G.~J. Stamatakis, O.~Simeone, and N.~Pappas, ``Optimizing information freshness over a channel that wears out,'' in \emph{Proceedings of the IEEE Asilomar Conference on Signals, Systems, and Computers}, 2023, pp. 85--89.

\bibitem{delfani2024semantics}
E.~Delfani and N.~Pappas, ``Semantics-aware status updates with energy harvesting devices: Query version age of information,'' in \emph{Proceedings of the IEEE International Symposium on Modeling and Optimization in Mobile, Ad Hoc, and Wireless Networks (WiOpt)}, 2024, pp. 177--184.

\bibitem{delfani2024optimizing}
------, ``Optimizing information freshness in constrained {IoT} systems: A token-based approach,'' \emph{IEEE Transactions on Communications (Early Access)}, 2024.

\bibitem{ayan2019age}
O.~Ayan, M.~Vilgelm, M.~Kl{\"u}gel, S.~Hirche, and W.~Kellerer, ``Age-of-information vs. value-of-information scheduling for cellular networked control systems,'' in \emph{Proceedings of the ACM/IEEE International Conference on Cyber-Physical Systems}, 2019, pp. 109--117.

\bibitem{agheli2024access}
P.~Agheli, N.~Pappas, and M.~Kountouris, ``Goal-oriented multiple access connectivity for networked intelligent systems,'' \emph{IEEE Communications Letters}, vol.~28, no.~8, pp. 1795--1799, 2024.

\bibitem{agheli2023semantic}
------, ``Semantic filtering and source coding in distributed wireless monitoring systems,'' \emph{IEEE Transactions on Communications}, vol.~72, no.~6, pp. 3290--3304, 2024.

\bibitem{bertsekas2007volume}
D.~Bertsekas, \emph{Dynamic programming and optimal control}.\hskip 1em plus 0.5em minus 0.4em\relax Athena Scientific, 2007, vol.~2.

\bibitem{altman1999constrained}
E.~Altman, \emph{Constrained {Markov} decision processes}.\hskip 1em plus 0.5em minus 0.4em\relax CRC Press, 1999.

\bibitem{bellman1957dynamic}
R.~Bellman, \emph{Dynamic Programming}.\hskip 1em plus 0.5em minus 0.4em\relax Princeton University Press, 1957.

\bibitem{puterman2014markov}
M.~L. Puterman, \emph{Markov decision processes: Discrete stochastic dynamic programming}.\hskip 1em plus 0.5em minus 0.4em\relax John Wiley \& Sons, 2014.

\bibitem{Wood2009}
G.~Wood, \emph{Bisection global optimization methods}, C.~A. Floudas and P.~M. Pardalos, Eds.\hskip 1em plus 0.5em minus 0.4em\relax Springer Science \& Business Media, 2009.

\bibitem{littman2013complexity}
M.~L. Littman, T.~L. Dean, and L.~P. Kaelbling, ``On the complexity of solving {Markov} decision problems,'' \emph{Proceedings of the Conference on Uncertainty in Artificial Intelligence (UAI)}, pp. 394--402, 1995.

\bibitem{DQN}
V.~Mnih, K.~Kavukcuoglu, D.~Silver, A.~Graves, I.~Antonoglou, D.~Wierstra, and M.~A. Riedmiller, ``Playing atari with deep reinforcement learning,'' \emph{arXiv:1312.5602}, 2013.

\bibitem{A2C}
V.~Mnih, A.~P. Badia, M.~Mirza, A.~Graves, T.~P. Lillicrap, T.~Harley, D.~Silver, and K.~Kavukcuoglu, ``Asynchronous methods for deep reinforcement learning,'' \emph{arXiv:1602.01783}, 2016.

\bibitem{PPO}
J.~Schulman, F.~Wolski, P.~Dhariwal, A.~Radford, and O.~Klimov, ``Proximal policy optimization algorithms,'' \emph{arXiv:1707.06347}, 2017.

\end{thebibliography}

\end{document}